\documentclass[final,5p,times,twocolumn]{elsarticle}

\usepackage[nolist,nohyperlinks]{acronym}
\usepackage{amsmath}
\usepackage{breqn}
\usepackage{balance}
\usepackage{framed}
\usepackage{siunitx}
\usepackage{soul}
\usepackage{multirow}
\usepackage{etoolbox}
\usepackage{algorithm}
\usepackage{algorithmic}

\usepackage{blindtext}

\usepackage{standalone}

\usepackage{tikz}
\usepackage{pgfplots}
\usetikzlibrary{spy}
\usepgfplotslibrary{groupplots}
\usetikzlibrary{backgrounds}
\usetikzlibrary{patterns}
\usepgfplotslibrary{fillbetween}
\pgfplotsset{compat=1.18}

\usepackage{siunitx}          
\usepackage[version=4]{mhchem} 

\definecolor{NZgreen}{RGB}{44,160,44}
\definecolor{SIblue}{RGB}{31,119,180}
\definecolor{JPred}{RGB}{214,39,40}
\definecolor{AoIBlue}{RGB}{31,119,180}
\definecolor{AoIRed}{RGB}{214,39,40}
\definecolor{AoIOrange}{RGB}{255,127,14}

\usepackage{lipsum} 
\usepackage{graphicx}
\usepackage{subcaption}
\usepackage{color}
\usepackage[table,xcdraw]{xcolor}

\graphicspath{{./images/}}

\usepackage{amsthm}
\usepackage{amsfonts}
\newtheorem{proposition}{Proposition}
\newtheorem{remark}{Remark}
\newtheorem{corollary}{Corollary}
\usepackage{pifont}

\usepackage{url}
\usepackage[bookmarks=false]{hyperref}
\hypersetup{
  hidelinks,
  colorlinks=false,
  breaklinks=true,
  bookmarksopen=false
}

\newtoggle{authorcomment}
\toggletrue{authorcomment}

\begin{acronym}[MACHU]
  \acro{iot}[IoT]{Internet of Things}
  \acro{iiot}[IIoT]{Industrial Internet of Things}
  \acro{cr}[CR]{Cognitive Radio}
  \acro{ofdm}[OFDM]{orthogonal frequency-division multiplexing}
  \acro{ofdma}[OFDMA]{orthogonal frequency-division multiple access}
  \acro{scfdma}[SC-FDMA]{single carrier frequency division multiple access}
  \acro{rbi}[RBI]{ Research Brazil Ireland}
  \acro{rfic}[RFIC]{radio frequency integrated circuit}
  \acro{sdr}[SDR]{Software Defined Radio}
  \acro{sdn}[SDN]{Software Defined Networking}
  \acro{su}[SU]{Secondary User}
  \acro{ra}[RA]{Resource Allocation}
  \acro{qos}[QoS]{quality of service}
  \acro{usrp}[USRP]{Universal Software Radio Peripheral}
  \acro{mno}[MNO]{Mobile Network Operator}
  \acro{mnos}[MNOs]{Mobile Network Operators}
  \acro{gsm}[GSM]{Global System for Mobile communications}
  \acro{tdma}[TDMA]{Time-Division Multiple Access}
  \acro{fdma}[FDMA]{Frequency-Division Multiple Access}
  \acro{gprs}[GPRS]{General Packet Radio Service}
  \acro{msc}[MSC]{Mobile Switching Centre}
  \acro{bsc}[BSC]{Base Station Controller}
  \acro{umts}[UMTS]{universal mobile telecommunications system}
  \acro{Wcdma}[WCDMA]{Wide-band code division multiple access}
  \acro{wcdma}[WCDMA]{wide-band code division multiple access}
  \acro{cdma}[CDMA]{code division multiple access}
  \acro{lte}[LTE]{Long Term Evolution}
  \acro{papr}[PAPR]{peak-to-average power rating}
  \acro{hn}[HetNet]{heterogeneous networks}
  \acro{phy}[PHY]{physical layer}
  \acro{mac}[MAC]{medium access control}
  \acro{amc}[AMC]{adaptive modulation and coding}
  \acro{mimo}[MIMO]{multiple input multiple output}
  \acro{rats}[RATs]{radio access technologies}
  \acro{vni}[VNI]{visual networking index}
  \acro{rbs}[RB]{resource blocks}
  \acro{rb}[RB]{resource block}
  \acro{ue}[UE]{user equipment}
  \acro{cqi}[CQI]{Channel Quality Indicator}
  \acro{hd}[HD]{half-duplex}
  \acro{fd}[FD]{full-duplex}
  \acro{sic}[SIC]{self-interference cancellation}
  \acro{si}[SI]{self-interference}
  \acro{bs}[BS]{base station}
  \acro{fbmc}[FBMC]{Filter Bank Multi-Carrier}
  \acro{ufmc}[UFMC]{Universal Filtered Multi-Carrier}
  \acro{scm}[SCM]{Single Carrier Modulation}
  \acro{isi}[ISI]{inter-symbol interference}
  \acro{ftn}[FTN]{Faster-Than-Nyquist}
  \acro{m2m}[M2M]{machine-to-machine}
  \acro{mtc}[MTC]{machine type communication}
  \acro{mmw}[mmWave]{millimeter wave}
  \acro{bf}[BF]{beamforming}
  \acro{los}[LOS]{line-of-sight}
  \acro{nlos}[NLOS]{non line-of-sight}
  \acro{capex}[CAPEX]{capital expenditure}
  \acro{opex}[OPEX]{operational expenditure}
  \acro{ict}[ICT]{Information and Communications Technology}
  \acro{sp}[SP]{service providers}
  \acro{inp}[InP]{infrastructure providers}
  \acro{mvnp}[MVNP]{mobile virtual network provider}
  \acro{mvno}[MVNO]{mobile virtual network operator}
  \acro{nfv}[NFV]{network function virtualization}
  \acro{vnfs}[VNF]{virtual network functions}
  \acro{cran}[C-RAN]{Cloud Radio Access Network}
  \acro{bbu}[BBU]{baseband unit}
  \acro{bbus}[BBU]{baseband units}
  \acro{rrh}[RRH]{remote radio head}
  \acro{rrhs}[RRH]{Remote radio heads} 
  \acro{sfv}[SFV]{sensor function virtualization}
  \acro{wsn}[WSN]{wireless sensor networks} 
  \acro{bio}[BIO]{Bristol is open}
  \acro{vitro}[VITRO]{Virtualized dIstributed plaTfoRms of smart Objects}
  \acro{os}[OS]{operating system}
  \acro{www}[WWW]{world wide web}
  \acro{iotvn}[IoT-VN]{IoT virtual network}
  \acro{mems}[MEMS]{micro electro mechanical system}
  \acro{mec}[MEC]{Mobile edge computing}
  \acro{coap}[CoAP]{Constrained Application Protocol}
  \acro{vsn}[VSN]{Virtual sensor network}
  \acro{rest}[REST]{REpresentational State Transfer}
  \acro{aoi}[AoI]{Age of Information}
  \acro{lora}[LoRa\texttrademark]{Long Range}
  \acro{iot}[IoT]{Internet of Things}
  \acro{snr}[SNR]{Signal-to-Noise Ratio}
  \acro{cps}[CPS]{Cyber-Physical System}
  \acro{uav}[UAV]{Unmanned Aerial Vehicle}
  \acro{rfid}[RFID]{Radio-frequency identification}
  \acro{lpwan}[LPWAN]{Low-Power Wide-Area Network}
  \acro{lgfs}[LGFS]{Last Generated First Served}
  \acro{wsn}[WSN]{wireless sensor network} 
  \acro{lmmse}[LMMSE]{Linear Minimum Mean Square Error}
  \acro{rl}[RL]{Reinforcement Learning}
  \acro{nb-iot}[NB-IoT]{Narrowband IoT}
  \acro{lorawan}[LoRaWAN]{Long Range Wide Area Network}
  \acro{mdp}[MDP]{Markov Decision Process}
  \acro{ann}[ANN]{Artificial Neural Network}
  \acro{dqn}[DQN]{Deep Q-Network}
  \acro{mse}[MSE]{Mean Square Error}
  \acro{ml}[ML]{Machine Learning}
  \acro{cpu}[CPU]{Central Processing Unit}
  \acro{ddpg}[DDPG]{Deep Deterministic Policy Gradient}
  \acro{ai}[AI]{Artificial Intelligence}
  \acro{gp}[GP]{Gaussian Processes}
  \acro{drl}[DRL]{Deep Reinforcement Learning}
  \acro{mmse}[MMSE]{Minimum Mean Square Error}
  \acro{fnn}[FNN]{Feedforward Neural Network}
  \acro{eh}[EH]{Energy Harvesting}
  \acro{wpt}[WPT]{Wireless Power Transfer}
  \acro{dl}[DL]{Deep Learning}
  \acro{yolo}[YOLO]{You Only Look Once}
  \acro{mec}[MEC]{Mobile Edge Computing}
  \acro{marl}[MARL]{Multi-Agent Reinforcement Learning}
  \acro{aoi}[AoI]{Age of Information}
  \acro{cf}[CF]{Carbon Footprint}
  \acro{ci}[CI]{Carbon Intensity}
  \acro{fcfs}[FCFS]{First Come First Served}
  \acro{lcfs}[LCFS]{Last Come First Served}
  \acro{qos}[QoS]{Quality of Service}
  \acro{snr}[SNR]{Signal-to-noise Ratio}
  \acro{saoi}[SAoI]{Sustainable Age of Information}
  \acro{aqi}[AQI]{age and quality of Information}
  \acro{voi}[VoI]{Value of Information}
  \acro{qaoi}[QAoI]{Query Age of Information}
  \acro{aoii}[AoII]{Age of Incorrect Information}
  \acro{wsn}[WSN]{Wireless Sensor Network}
  \acro{cwsaoi}[CWSAoI]{Cumulative Weighted Sum Age of Information}
  \acro{iwsaoi}[IWSAoI]{Instantaneous Weighted Sum Age of Information}
  \acro{mtu}[MTU]{Maximum Transmission Unit}
  \acro{cad}[CAD]{Channel Activity Detection}
  \acro{saoithe}[SAOITHE]{Sustainable Age-of-Information Scheduling for IoT Hardware-Constrained Edge Networks}
  \acro{dp}[DP]{Dynamic Programming}
  \acro{rmab}[RMAB]{Restless Multi-Armed Bandit problem}
  \acro{flops}[FLOPs]{Floating-Point Operations Per Second}
  \acro{isac}[ISAC]{Integrated Sensing and Communication}
  \acro{cav}[CAV]{Connected Autonomous Vehicles}
  \acro{swipt}[SWIPT]{Simultaneous Wireless Information and Power Transfer}
\end{acronym}

\journal{ Elsevier in 2026.}

\begin{document}

\begin{frontmatter}

\title{Carbon-Aware Age of Information Scheduling in IoT Networks}


\author[inst1]{Shih-Kai Chou}
\affiliation[inst1]{organization={Ericsson AB},
            city={Kista},
            country={Sweden}}

\author[inst2]{Maice Costa}
\affiliation[inst2]{organization={Teledyne},
  state={Virginia},
  country={USA}
}

\author[inst3]{Mihael Mohor\v{c}i\v{c}}
\affiliation[inst3]{organization={Jožef Stefan Institute},
  city={Ljubljana},
  country={Slovenia}}


\author[inst3]{Jernej Hribar}

\begin{abstract}
Ensuring timely information delivery while reducing environmental impact is a fundamental challenge in future large-scale 6G and beyond Internet of Things (IoT) deployments. Information timeliness is typically measured using the Age of Information (AoI) metric, whereas environmental impact is often characterized by the Carbon Footprint (CF), which depends on both energy consumption and the Carbon Intensity (CI), defined as the amount of CO$_2$-equivalent emissions produced per unit of consumed energy. Because CI varies over time, minimizing energy consumption alone does not necessarily minimize CF: transmitting the same status update at different times may result in different carbon costs. To address this challenge, we propose Sustainable Age-of-Information Scheduling for IoT Hardware-Constrained Edge Networks (SAOITHE), a scalable scheduling heuristic inspired by Whittle-index decomposition. A closed-form index is derived for a relaxed infinite-horizon single-source model with immediate transmission and unit AoI reset; this index is subsequently used as a priority rule for the original finite-horizon, multi-source system with buffered gateway and server dynamics. Using real-world CI traces across low-, medium-, and high-CI regions, the proposed heuristic operates within the allocated CF and duty-cycle budgets in all evaluated configurations and achieves lower average AoI than baseline policies. The AoI reduction relative to Round Robin is approximately $25\%$ and $20\%$ in low- and medium-CI regions, respectively, and up to $75\%$ in high-CI settings, while maintaining quasilinear per-slot computational complexity.
\end{abstract}




\fntext[chou-note]{This work was performed while S.-K. Chou was with the
Jo\v{z}ef Stefan Institute. S.-K. Chou is currently with Ericsson AB, Kista,
Sweden. The views expressed in this paper are those of the author and do not
necessarily reflect the views of Ericsson AB.}

\fntext[costa-note]{M. Costa is with Teledyne, USA. This work was unrelated
to her current affiliation and does not reflect the views of Teledyne
Technologies.}

\cortext[cor1]{Corresponding author.}
\ead{jernej.hribar@ijs.si}

\end{frontmatter}



%
\section{Introduction}
\label{sec:intro}
With the rise of data-driven and latency-sensitive applications, such as \ac{cav}, \ac{iiot}, remote healthcare, and \ac{isac}, transmitting fresh and up-to-date information is critical~\cite{akyildiz20206g}. For such systems, the \ac{aoi} metric~\cite{kaul2012real} has been adopted not only to measure the timeliness of collected information, but also to support the design of scheduling policies for real-time systems~\cite{yates2021age}. Scheduling data collection can have a crucial impact on how sustainable, i.e., environmentally friendly, future networks can become, as the \ac{cf} depends not only on the amount of energy required to collect information, but also on how the consumed energy was generated~\cite{yang2025toward, chou2025energy}. The greater the share of renewable energy sources used for electricity generation, the lower the \ac{cf} of the consumed energy. 
Therefore, when considering the sustainability of 6G and beyond networks, where international standardization bodies have explicitly identified environmental sustainability as a core Key Performance Indicator (KPI) \cite{itu2023imt2030}, scheduling, i.e., deciding when information is transmitted, captured, and processed, becomes crucial.


Typically, the \ac{cf} of an \ac{ict} system is estimated by associating the energy consumed for data transmission with the corresponding \ac{ci} per unit of energy~\cite{CFreport}. \ac{ci} is defined as the amount of CO$_2$-equivalent (CO$_2$eq) emissions produced per unit of consumed energy and varies over time and geographical region. For example, Fig.~\ref{fig:intro} shows the variation of \ac{ci} over a day for three different \ac{ci} regions using \ac{ci} traces~\cite{electricitymaps}. Consequently, the same amount of energy required to transmit information can result in a relatively low \ac{cf} if data transmission and processing occur during periods of low \ac{ci}, whereas the \ac{cf} can be three to four times higher when the \ac{ci} is high. Moreover, \ac{ci} cannot be directly controlled because it depends on the time-varying electricity generation mix and is influenced by factors such as demand, pricing, dispatch rules, and imports. Therefore, reducing the \ac{cf} of data collection requires deciding when information is generated and processed to minimize its environmental impact.

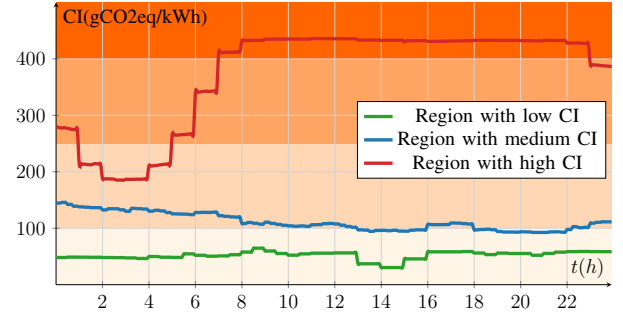
\begin{figure}
	\centering
    \Large \begin{tikzpicture}
\definecolor{orange1}{RGB}{255, 245, 230} 
\definecolor{orange2}{RGB}{255, 215, 180} 
\definecolor{orange3}{RGB}{255, 165, 100} 
\definecolor{orange4}{RGB}{255, 100, 0}   

\definecolor{NZgreen}{RGB}{44,160,44}
\definecolor{SIblue}{RGB}{31,119,180}
\definecolor{JPred}{RGB}{214,39,40}

\large
\begin{axis}[
	axis lines=middle,
	xlabel= $t(h)$, 
        ylabel=CI(gCO2eq/kWh),
	grid=both,
	major grid style={gray!35},
        ytick ={100, 200, 300, 400},
        xtick ={0,2,4,6,8,10,12,14,16,18, 20, 22},
        scale ticks above exponent={2},
        ymin=0,ymax=500,
	width=1.5\linewidth,
	height=0.85\linewidth,
	legend style={at={(0.99,0.65)}},
]

\begin{scope}[on background layer]
  \fill[orange1,opacity=10] ({rel axis cs:0,0}) rectangle ({rel axis cs:1.0,0.2});
    \fill[orange2,opacity=10] ({rel axis cs:0,.2}) rectangle ({rel axis cs:1.0,0.5});
     \fill[orange3,opacity=10] ({rel axis cs:0,.5}) rectangle ({rel axis cs:1.0,0.8});
     \fill[orange4,opacity=10] ({rel axis cs:0,.8}) rectangle ({rel axis cs:1.0,1.0});
\end{scope}

\addplot[line width=2pt, color=NZgreen,  smooth, mark options={solid}]table[x=Time_Hours,y=NZ_CI_gCO2kWh,col sep=comma]{csv_data/transmission_timing_edge.csv};

\addplot[line width=2pt, color=SIblue,  smooth, mark options={solid}]table[x=Time_Hours,y=SI_CI_gCO2kWh,col sep=comma]{csv_data/transmission_timing_edge.csv};

\addplot[line width=2pt, color=JPred,  smooth, mark options={solid}]table[x=Time_Hours,y=JP_CI_gCO2kWh,col sep=comma]{csv_data/transmission_timing_edge.csv};



\addlegendentry{Region with low CI};
\addlegendentry{Region with medium CI};
\addlegendentry{Region with high CI};


\end{axis}

\end{tikzpicture}
         \caption{Daily variation in CI for regions with low, medium, and high CI over a day.}
	\label{fig:intro}
    \vspace{-15pt}
\end{figure}

For many systems, especially large-scale deployments using \acp{lpwan} such as \ac{lorawan}, the \ac{aoi} has been shown to be a suitable metric for characterizing the freshness of information~\cite{hribar2020energy}. Although many \ac{lpwan} applications are delay-tolerant, information freshness remains important in monitoring-oriented \ac{iot} use cases such as industrial condition monitoring, fault or leak detection, asset management, and environmental sensing, where stale measurements may postpone detection or operational intervention~\cite{abbas2023aoi,bonilla2023lorawan}.
In such systems, maintaining fresh information is fundamentally a scheduling problem, since status updates cannot be transmitted arbitrarily often. This is due to additional constraints, such as duty-cycle regulations and restricted channel access at the gateway. As a result, the scheduler must decide carefully when information sources should transmit. Moreover, as the number of information sources increases, the complexity of the scheduling problem also grows, which makes it increasingly difficult to find a suitable solution.


In the \ac{aoi} literature, such scheduling problems are often modeled as Markov Decision Processes (MDPs) with resource budgets, but the state space grows rapidly with the number of information sources. To overcome this, we employ a Whittle-index-inspired approach, originally proposed to reduce the complexity of restless bandit problems~\cite{whittle1988restless}. This methodology is widely adopted in \ac{aoi} literature to decompose intractable multi-source problems into manageable subproblems and prioritize transmissions based on their urgency~\cite{Hsu_AoI_earlywork,aoi_whittle_7,whittle_aoi_0}.

The main contributions of this paper are as follows:
\begin{itemize}
    \item We model the \ac{cf}-aware status-updating problem as a finite-horizon \acs{mdp} with \ac{cf} and duty-cycle budgets and channel-capacity constraints, and seek to minimize a quadratic \ac{aoi} penalty. The resulting exact scheduling problem is intractable for large numbers of sources, which motivates a scalable heuristic solution.
    \item We introduce the \ac{saoithe} framework, which leverages a Whittle-index-based scheduling heuristic to obtain scalable policies for the formulated optimization problem. Building on established ideas in the \ac{aoi} and restless-bandit literature, the framework constructs a practical policy by (i) relaxing the cumulative constraints into time-averaged penalties, (ii) analyzing the associated infinite-horizon, single-source subproblem via a Whittle-index formulation, and (iii) using the resulting index as a priority rule for the original finite-horizon, multi-source system with buffered gateway and server dynamics. 
    \item We evaluate the proposed framework using real-world \ac{ci} traces across low-, medium-, and high-\ac{ci} regions. The results show that \ac{saoithe} remains within the allocated \ac{cf} budget while achieving a lower average \ac{aoi} compared with baseline solutions. Furthermore, we demonstrate the \ac{aoi} reduction over the Round Robin baseline under the same \ac{cf} budget is approximately~$25\%$ and~$20\%$ in low- and medium-\ac{ci} regions, respectively, and increases to up to~$75\%$ in high-\ac{ci} settings, where blind schedulers incur a higher carbon cost. 
\end{itemize}




Next, we discuss the related work.




\section{Related Work}
\label{sec:related}


This work is related to studies that extend the concept of \ac{aoi} by incorporating additional network conditions and application requirements~\cite{yin2019only, rajaraman2021not, VOI_ref, maatouk2020age}. Metrics such as Effective \ac{aoi}~\cite{yin2019only} extend the baseline definition by incorporating semantic and correctness aspects of received data, while others adapt \ac{aoi} to various queuing models, scheduling policies, and update mechanisms. Metrics like Age-Quality Information (AQI)\cite{rajaraman2021not} and Value of Information (VoI)\cite{VOI_ref} jointly optimize timeliness, data quality, and energy efficiency, often employing a utility-based function that also encompasses \ac{aoi}. Similarly, the Age of Incorrect Information (AoII) metric~\cite{maatouk2020age} accounts for the correctness of updates and supports the design of optimal transmission policies under unreliable channel conditions and power constraints. In contrast to the aforementioned \ac{aoi}-based metrics, our work proposes a scheduler that minimizes \ac{aoi} while accounting for sustainability through \ac{cf}, thereby considering the environmental aspects that have been largely overlooked in prior studies.

This work also builds on research that explored the role of \ac{aoi} in enhancing energy efficiency~\cite{HUANG202329,xu2020info,zhang2023TWC,Hatami2021ToC,zhang2024TOC}. For example, the authors in~\cite{HUANG202329} applied \ac{rl} and system-level strategies to jointly optimize \ac{aoi} and energy efficiency for computation offloading in \ac{iiot}, supported by queuing models. Freshness in caching-enabled \ac{iot} networks was addressed in~\cite{xu2020info}, while the approach proposed in~\cite{zhang2023TWC} leveraged \ac{mec} in a \ac{wsn} for environmental monitoring to manage energy and computational constraints. In~\cite{Hatami2021ToC}, the authors focused on energy-harvesting \ac{iot} systems with on-demand updates, while in~\cite{zhang2024TOC} they explored \ac{uav}-assisted data collection, revealing trade-offs between data freshness and energy cost. However, none of these works investigated the environmental impact and its connection to \ac{aoi}, as considered in this paper. Moreover, while minimizing energy consumption is a necessary step toward reducing the \ac{cf} of a system, it is not sufficient on its own. To address this gap, in this work, we propose a solution that minimizes the environmental impact of the system by also taking into account the temporal variation of \ac{ci}.

On the other hand, the Whittle index~\cite{whittle1988restless} has been widely adopted to design transmission scheduling solutions that minimize \ac{aoi} in the system~\cite{Hsu_AoI_earlywork,aoi_whittle_7,aoi_whittle_6,whittle_aoi_0,aoi_whittle_3,aoi_whittle_8}. For example, the authors in~\cite{Hsu_AoI_earlywork} applied Whittle's framework for restless bandits to design a transmission scheduler that minimizes \ac{aoi}. In~\cite{aoi_whittle_7}, the authors demonstrated that Whittle index scheduling policies result in near-optimal performance for minimizing \ac{aoi}, while the authors in~\cite{aoi_whittle_6} showed that the Whittle policy has analytically provable optimality in the many-user regime. Similarly, in~\cite{whittle_aoi_0}, the authors provided a proof of optimality for a low-complexity Whittle-index scheduling policy for multi-source status updating systems. The low-complexity Whittle index was also employed in~\cite{aoi_whittle_3}, where the authors proposed Channel-Aware Age of Information (CA-AoI), which considers the channel conditions at the source in a multi-source system. In addition, the authors in~\cite{aoi_whittle_8} employ the Whittle index method to design an energy-aware scheduling policy for sources that minimize \ac{aoi} in the system. However, none of these papers consider the \ac{ci} and the resulting \ac{cf} in connection with \ac{aoi}, as proposed in our work.

This work also builds on our prior work~\cite{chou2026towards}, in which we analyzed a system with one information source. Our results showed that \ac{cf} grows with the number of status updates from the source, i.e., the update rate, whereas the gain in average \ac{aoi} is generally nonlinear and may exhibit either diminishing returns or even degradation due to queuing and other delays in the system. Therefore, scheduling updates from sources so as to minimize the average \ac{aoi} while satisfying a constraint on the total available \ac{cf} budget becomes, in general, a more complex problem than the energy-minimization problem, since the \ac{ci} also varies over time. In this paper, we build on our prior findings and propose a scalable scheduling solution, \ac{saoithe}, designed to reduce \ac{aoi} while empirically satisfying the specified \ac{cf} requirements.

\section{System Model}
\label{sec:system_model}

\begin{figure}[!t]
    \centering
    \includegraphics[width=0.5\textwidth]{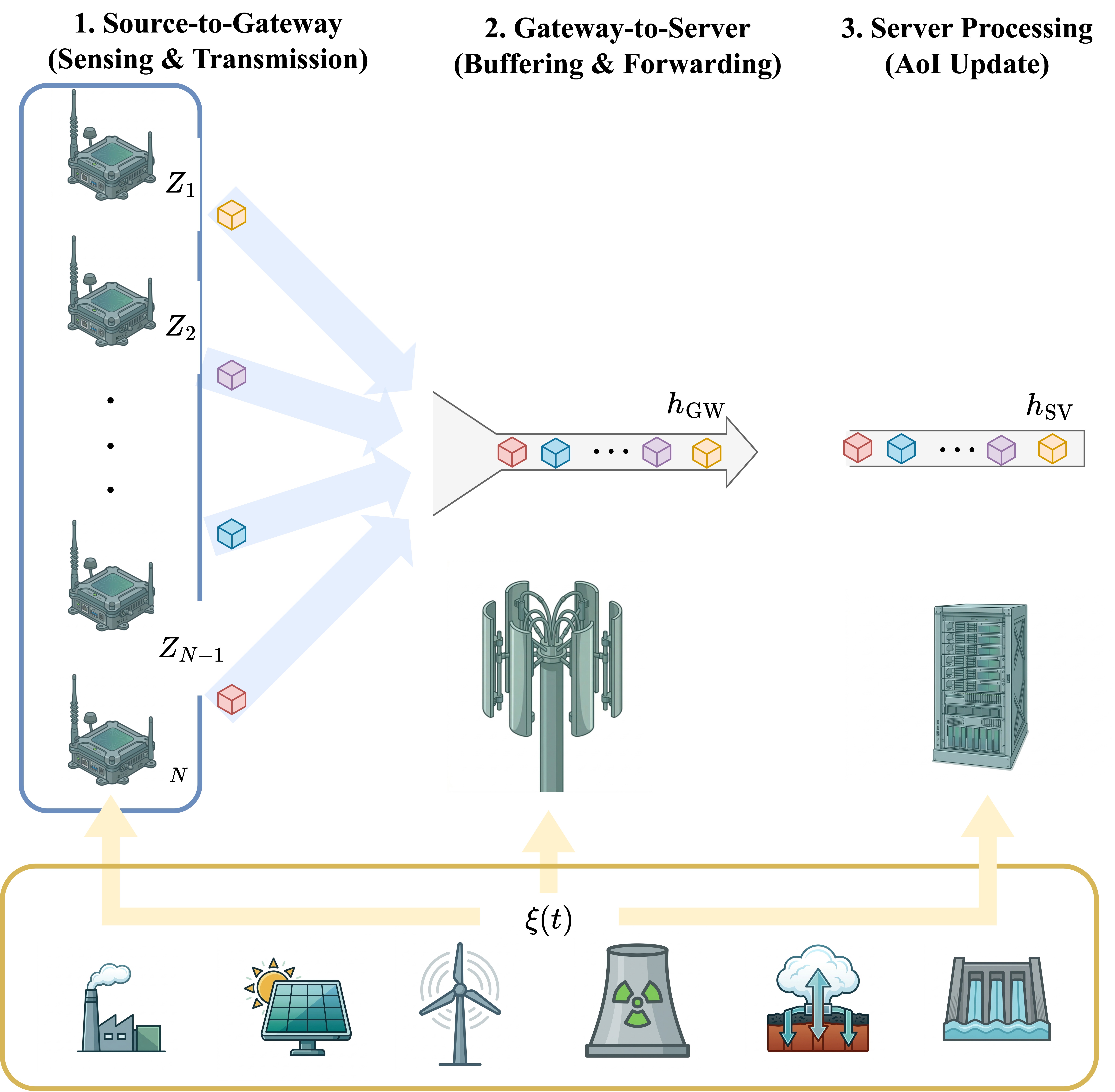}
    \caption{Considered system model with $N$ sources transmitting status updates through a gateway to a server.}
    \label{fig:system_model}
\end{figure}

Let us consider a system of $N$ sensors, i.e., information sources, denoted as $Z_n$, where $n \in \{1,2,\ldots,N\}$, that transmit status updates through a gateway to a server, as illustrated in Fig.~\ref{fig:system_model}. Such a system, in which information is generated by hardware-constrained devices at the edge and transmissions are often scheduled, is typical of \ac{lpwan} based on technologies such as \ac{lorawan}, Sigfox, or Mioty. The system operates over a finite time horizon of length $T$, which is discretized into $T$ decision intervals (slots) indexed by $t \in \{1,\ldots,T\}$. Each generated status update contains a measurement along with a timestamp indicating when it was generated. Therefore, we define the \ac{aoi} for the $n$-th source as follows:

\begin{equation}
\Delta_n(t) = t - z_n(t),
\label{eq:inst_aoi}
\end{equation}
\noindent where $z_n(t)$ is the generation time slot of the most recently received status update from source $n$ available at the server at time slot $t$.


To model the sequential decision-making problem of scheduling status updates over time, we formulate the considered system as a \ac{mdp}.

\textbf{State space}: We define the joint state space for the considered system model as follows:
\begin{equation}
\begin{aligned}
    \mathcal{S}_{t}= &\left(\boldsymbol{\Delta}(t), \mathbf{b}_{\mathrm{DEV}}(t), \boldsymbol{\tau}_{\mathrm{DEV}}(t),\right.\\
    & \left.
    \mathbf{h}_{\mathrm{GW}}(t), \mathbf{b}_{\mathrm{GW}}(t), \boldsymbol{\tau}_{\mathrm{GW}}(t), \mathbf{h}_{\mathrm{SV}}(t)\right),
    \label{eq:state_general}
\end{aligned}
\end{equation}
\noindent where $\boldsymbol{\Delta}(t) = [\Delta_1(t), \ldots, \Delta_N(t)]$ is a vector of \ac{aoi} values for all information sources in the system. Moreover, $\mathbf{b}_{\mathrm{DEV}}(t)$ and $\boldsymbol{\tau}_{\mathrm{DEV}}(t)$ are vectors tracking the busy status and the remaining transmission time for each source, respectively. The terms $\mathbf{b}_{\mathrm{GW}}(t)$ and $\boldsymbol{\tau}_{\mathrm{GW}}(t)$ represent the corresponding forwarding status and timer for the gateway. Furthermore, $\mathbf{h}_{\mathrm{GW}}(t)$ and $\mathbf{h}_{\mathrm{SV}}(t)$ denote the buffered status update ages in the gateway and server buffers, respectively. 

From the perspective of the $n$-th source, the binary variables $b_{\mathrm{DEV},n}(t)$ and $b_{\mathrm{GW},n}(t)$ indicate the busy status of the device and gateway, respectively. Specifically, when $b_{\mathrm{DEV},n}(t)=1$ or $b_{\mathrm{GW},n}(t)=1$, the corresponding entity is in active mode, whereas $0$ indicates idle mode. The variables $\tau_{\mathrm{DEV},n}(t)$ and $\tau_{\mathrm{GW},n}(t)$ track the remaining time in the current transmission and forwarding cycle, respectively. Finally, $h_{\mathrm{GW},n}(t)$ and $h_{\mathrm{SV},n}(t)$, with values in $\{0, \dots, \delta_{\max}\}$, denote the \ac{aoi} of the packet currently residing in the gateway and server buffers, respectively.

\textbf{Actions}:
At each time-step, the system has to coordinate the activity of all network entities, i.e., sources, gateway, and server. As a result, we define the joint action space as $\mathcal{A}_{t}=\mathcal{A}_{\mathrm{DEV}}\times\mathcal{A}_{\mathrm{GW}}\times\mathcal{A}_{\mathrm{SV}}$, with each action $\mathbf{a}(t)\in\mathcal{A}_{t}$ given by:
\begin{equation}
    \mathbf{a}(t) = \bigl(\mathbf{a}_{\mathrm{DEV}}(t), \mathbf{a}_{\mathrm{GW}}(t), \mathbf{a}_{\mathrm{SV}}(t)\bigr),
    \label{eq:action_general}
\end{equation}

\noindent where $\mathbf{a}_{\mathrm{DEV}}(t)$ is a binary vector indicating which of the $N$ sources initiate a transmission cycle, i.e., $\mathbf{a}_{\mathrm{DEV}}(t) = [a_{\mathrm{DEV},1}(t), \ldots, a_{\mathrm{DEV},N}(t)]^{T}$. Similarly, $\mathbf{a}_{\mathrm{GW}}(t)$ and $\mathbf{a}_{\mathrm{SV}}(t)$ control the forwarding and processing decisions for the gateway and server, respectively.

From the perspective of an $n$-th source in the system, each action element is a binary decision variable, i.e., $a_{\mathrm{DEV},n}(t)$, $a_{\mathrm{GW},n}(t)$, and $a_{\mathrm{SV},n}(t) \in \{0,1\}$. This means that, when $a_{\mathrm{DEV},n}(t)=1$, the device initiates sensing and transmission, while $a_{\mathrm{GW},n}(t)=1$ and $a_{\mathrm{SV},n}(t)=1$ indicate that the gateway forwards and the server processes the status update, respectively. In this work, the scheduling decision is fully determined by $\mathbf{a}_{\mathrm{DEV}}(t)$; given this device activation vector and the current system state, $\mathbf{a}_{\mathrm{GW}}(t)$ and $\mathbf{a}_{\mathrm{SV}}(t)$ follow deterministically from the store-and-forward dynamics of the gateway and server.

Furthermore, due to \ac{lpwan} constraints, an action is feasible only if the corresponding entity is idle, i.e., $b_{\mathrm{DEV},n}(t)=0$ and $b_{\mathrm{GW},n}(t)=0$, and, for forwarding or processing, if a packet is available in the respective buffer, i.e., $h_{\mathrm{GW},n}(t) > 0$ or $h_{\mathrm{SV},n}(t) > 0$. Throughout this work, we therefore restrict attention to scheduling policies that operate over this feasible action set: a transmission can only be initiated when the corresponding device and gateway are idle, and forwarding or processing can only occur when a packet is present in the respective buffer. At the system level, the number of simultaneous transmissions is limited by the available channel capacity, which we denote by $M$. Consequently, the feasible transmission actions at every time slot $t$ are constrained by:
\begin{equation}
 \sum_{n=1}^{N} a_{\mathrm{DEV},n}(t)\leq M,
 \label{eq:action_constraint}
\end{equation}
\noindent where $a_{\mathrm{DEV},n}(t)$ represents the binary transmission decision for the $n$-th source.






\textbf{Transition Function}: The system dynamics are governed by the physical constraints of a typical \ac{lpwan} system, i.e., store-and-forward logic for status updates. In particular, while a generated status update waits in a buffer before being transmitted, its age continues to increase, which prevents the system from arbitrarily delaying transmissions to reduce \ac{cf} without incurring an \ac{aoi} penalty at the destination. Accordingly, the vector of destination-side \ac{aoi} values, $\boldsymbol{\Delta}(t)$, evolves such that the age of the $n$-th source is reset not to zero, but to the age of the processed status update upon reception. Otherwise, the \ac{aoi} increases linearly:
\begin{equation}
    \Delta_{n}(t+1)=(\Delta_{n}(t)+1)(1-\delta_{n}(t))+(h_{\mathrm{SV},n}(t)+1)\delta_{n}(t),
    \label{eq:aoi_general_tr}
\end{equation}
\noindent where $\delta_n(t) = a_{\mathrm{SV}, n}(t) \cdot \mathbb{I}(h_{\mathrm{SV}, n}(t) > 0)$ indicates that a valid packet from the $n$-th device has been processed and delivered during time slot $t$, and $h_{\mathrm{SV},n}(t)$ denotes the \ac{aoi} of that packet at the time of reception. Reflecting the three-entity architecture of the considered system, the transition function of the $n$-th source can be decomposed into three distinct stages:

\textbf{1. Source-to-Gateway (Sensing \& Transmission):}
The state of the $n$-th source $(b_{\mathrm{DEV},n}, \tau_{\mathrm{DEV},n})$ evolves based on the action $a_{\mathrm{DEV},n}(t)$, i.e., a transmission cycle ($b_{\mathrm{DEV},n}=1$) decrements the timer $\tau_{\mathrm{DEV},n}$. Here, we first define $I_{\mathrm{arr},n}(t)$ as the uplink completion indicator, which is expressed as:

\begin{equation}
    I_{\mathrm{arr},n}(t) = 
    \begin{cases} 
        1, & \text{if } b_{\mathrm{DEV},n}(t)=1 \wedge \tau_{\mathrm{DEV},n}(t)=1, \\
        0, & \text{otherwise}.
    \end{cases}
    \label{eq:uplink_indicator}
\end{equation}

\textbf{2. Gateway-to-Server (Buffering \& Forwarding):}
The gateway buffer state for the $n$-th source $h_{\mathrm{GW},n}(t)$ represents the \ac{aoi} of the packet. If a new packet arrives ($I_{\mathrm{arr},n}=1$), it enters the buffer with an initial \ac{aoi} of $1$ (representing the transmission slot). If the gateway holds the packet, its \ac{aoi} increments. If the gateway forwards the packet ($a_{\mathrm{GW},n}=1$), the buffer clears. Therefore, we can express the evolution of gateway buffer state ($h_{\mathrm{GW},n}(t+1)$) as follows:
\vspace{-5pt}
\begin{equation}
\begin{aligned}
        h_{\mathrm{GW},n}&(t+1) =\\ 
    &\begin{cases} 
        1, & \text{if } I_{\mathrm{arr},n}(t) = 1, \\
        0, & \text{if } a_{\mathrm{GW},n}(t) = 1 \wedge h_{\mathrm{GW},n}(t) > 0, \\
        h_{\mathrm{GW},n}(t) + 1, & \text{if } h_{\mathrm{GW},n}(t) > 0, \\
        0, & \text{otherwise}.
    \end{cases}
\end{aligned}
\label{eq:h_GW_t_1}
\end{equation}

\noindent The server buffer $h_{\mathrm{SV}}(t)$ evolves similarly. If the gateway completes a forwarding cycle ($\tau_{\mathrm{GW}}=1$), the packet arrives at the server, carrying its accumulated \ac{aoi} from the gateway plus the forwarding time. Therefore, the evolution of $h_{\mathrm{SV}}(t)$ can be derived as:

\vspace{-5pt}
\begin{equation}
\begin{aligned}
    h_{\mathrm{SV},n}&(t+1) =\\ 
&    \begin{cases} 
        h_{\mathrm{GW},n}(t) + T_{\mathrm{fwd}}, & \text{if } \tau_{\mathrm{GW},n}(t) = 1, \\
        0, & \text{if } a_{\mathrm{SV},n}(t) = 1 \wedge h_{\mathrm{SV},n}(t) > 0, \\
        h_{\mathrm{SV},n}(t) + 1, & \text{if } h_{\mathrm{SV},n}(t) > 0, \\
        0, & \text{otherwise}.
    \end{cases}
\end{aligned}
\label{eq:h_SV_t_1}
\end{equation}

\textbf{3. Server Processing (AoI Update):}
Finally, the global \ac{aoi} $\boldsymbol{\Delta}(t)$ is updated based on the server's processing actions. More specifically, for each information source $n\in\{1,\ldots, N\}$, the corresponding AoI value $\Delta_{n}(t)$ resets to the packet's current \ac{aoi} $h_{\mathrm{SV},n}(t)$ plus the processing slot, reflecting the true staleness of the information:
\begin{equation}
\Delta_{n}(t+1) = 
\begin{cases} 
h_{\mathrm{SV},n}(t) + 1, & \text{if } \delta_n(t) = 1, \\ 
\Delta_{n}(t) + 1, & \text{otherwise}.
\end{cases}
\label{eq:aoi_transition}
\end{equation}

\noindent While the system evaluation tracks the true, unbounded AoI ($\Delta_n$) to accurately reflect real-world staleness excursions, the theoretical \ac{mdp} state space strictly truncates all age variables—including the device AoI ($\Delta_n$) and the buffered ages ($h_{\mathrm{GW},n}, h_{\mathrm{SV},n}$) at a maximum value of $\delta_{\max}$. This state-space truncation is a standard mathematical necessity to ensure a finite state space for the derivation of the Dynamic Programming baseline and the Whittle index policies.

\textbf{Cost Function}:
To evaluate the system performance, we define the immediate cost function, $\mathcal{C}(\mathcal{S}_t, \mathcal{A}_t)$, primarily based on the freshness of the information. We adopt a quadratic staleness penalty to strongly penalize large peaks in \ac{aoi}, ensuring the reliability of status updates. As a result, it can be expressed as:


\begin{equation}
    \mathcal{C}(\mathcal{S}_{t},\mathcal{A}_{t}) = \sum_{n=1}^{N} \Delta_n^{2}(t).
    \label{eq:cost_function}
\end{equation}


\noindent Simultaneously, we track the environmental impact of the system using \ac{cf} as a key metric~\cite{CFreport}. For the $n$-th source, the \ac{cf} depends on the time-varying \ac{ci}, denoted by $\xi(t)$, and on the energy consumed by that source over time. The \ac{ci} reflects how carbon-intensive the energy source is at a given time, as discussed in~\cite{trihinas2022towards}. We adopt the standard measure of \ac{ci} as the amount of CO$_2$-equivalent emissions produced per kilowatt-hour of consumed energy. 

\section{Problem Formulation}
\label{sec:problem_formulation}




Having defined the considered system with an \ac{mdp}, we now formulate the corresponding optimization problem together with the associated \ac{cf} constraint.
The main objective is to obtain a scheduling policy, $\pi^{\ast}$, that minimizes the cumulative staleness penalty (in the form of \ac{aoi}) over a finite operational horizon of $T$ time slots, e.g., a day, 
while satisfying the environmental and regulatory constraints, i.e., an allocated \ac{cf} and duty-cycle budgets. First, we define the total cost over the considered time horizon $T$ as follows:


\begin{equation}
    J_{\Delta}(\pi) = \sum_{t=1}^{T} \sum^{N}_{n=1}\Delta^{2}_{n}(t).
    \label{eq:cost_aoi}
\end{equation}


The quadratic penalty $f(\Delta_n(t))=\Delta_n^2(t)$ is adopted to assign greater importance to more stale information. Because the AoI is inherently non-negative ($\Delta_{n} (t)\geq 0, \forall t$), the quadratic function preserves the strict monotonicity of a linear penalty. While the scalar value of the cumulative objective differs, the fundamental direction of the optimization is maintained, as the state that minimizes the immediate penalty remains identical ($\arg \min \Delta (t) = \arg \min \Delta^{2} (t)$). Furthermore, its marginal increase, $f(\Delta+1)-f(\Delta)=2\Delta+1$, grows with the current \ac{aoi}, meaning that an additional delay is penalized more strongly when a source has already remained stale for a long period. This property encourages the scheduler to avoid prolonged freshness degradation and to distribute update opportunities more evenly across sources. The resulting objective $J_{\Delta}(\pi)$ is mathematically distinct from the conventional time-averaged \ac{aoi} objective. In particular, the second moment satisfies $\mathbb{E}[\Delta^2]=(\mathbb{E}[\Delta])^2+\operatorname{Var}(\Delta)$, showing that it reflects both the average \ac{aoi} and its variability. Consequently, minimizing the quadratic \ac{aoi} penalty places additional emphasis on reducing large \ac{aoi} excursions and long tails of staleness, rather than treating all unit increases in \ac{aoi} value equally \cite{whittle_aoi_0}.

Similarly, the total cumulative \ac{cf} ($J_{\mathrm{C}}(\pi)$) and the total transmission duty usage ($J_{\mathrm{D}}(\pi)$) over the horizon are given by:

\begin{equation}
    J_{\mathrm{C}}(\pi) = \sum_{t=1}^{T}\sum_{n=1}^{N} \xi(t) \cdot E_{\mathrm{tot},n}(\mathcal{S}_t, \mathcal{A}_t),
    \label{eq:const_CF}
\end{equation}

and

\begin{equation}
    J_{\mathrm{D}}(\pi) = \sum_{t=1}^{T}\sum^{N}_{n=1} a_{\mathrm{DEV},n}(t) \cdot C_{\mathrm{duty}}, 
    \label{eq:const_tx}
\end{equation}

\noindent respectively, where $C_{\mathrm{duty}}$ represents the normalized airtime cost of a single transmission relative to the decision time step size $t$.




Finally, the constrained optimization problem is formulated as the minimization of the cumulative quadratic \ac{aoi} penalty
subject to constraints on the cumulative \ac{cf}, transmission duty usage, and the channel capacity $M$, and is expressed as follows:

\begin{equation}
\begin{aligned}
    \mathbf{P1}: \quad \min_{\pi} \quad & J_{\Delta}(\pi) \\
    \textrm{s.t.} \quad & J_{\mathrm{C}}(\pi) \leq \kappa, \\
                        & J_{\mathrm{D}}(\pi) \leq D,\\
                        & \sum^{N}_{n=1}a_{\mathrm{DEV},n}(t)\leq M,
\end{aligned}
\label{eq:p1}
\end{equation}

\noindent where $\kappa$ represents the total \ac{cf} budget allocated for the operational horizon, and $D$ denotes the aggregate maximum allowable transmission time for the network, derived from per-device regulatory duty-cycle limits, e.g., 1\% of the horizon $T$ per device, yielding a total network budget of $N \times 0.01T$. Such a budget constitutes a typical aggregate proxy for per-device duty-cycle limits for homogeneous \ac{lpwan} devices.
Additionally, it is important to note that, because the environmental and hardware constraints ($J_C(\pi) \leq \kappa$ and $J_D(\pi) \leq D$) are formulated as cumulative budgets evaluated over the entirety of the finite horizon $T$, rather than instantaneous per-slot limits, they constitute single scalar inequalities. Consequently, the dual problem does not require optimization over a time-dependent function space (e.g., $\lambda(t)$). Instead, single scalar dual variables are sufficient to penalize the violation of the total respective budgets. 



The problem $\mathbf{P1}$ is a finite-horizon scheduling problem with \ac{cf} and duty-cycle budgets and channel-capacity constraints. We introduce non-negative Lagrange multipliers, $\lambda$ and $\mu$, corresponding to the \ac{cf} budget and duty-cycle constraints, and define the Lagrangian $\mathcal{L}(\pi, \lambda, \mu)$ as the sum of the cumulative staleness cost and penalized constraint violations:

\begin{equation}
    \mathcal{L}(\pi, \lambda, \mu) = J_{\Delta}(\pi) + \lambda \left( J_{\mathrm{C}}(\pi) - \kappa \right) + \mu \left( J_{\mathrm{D}}(\pi) - D \right).
    \label{eq:lagrangian_0}
\end{equation}

\noindent By substituting the cumulative definitions from Eqs.~\eqref{eq:cost_aoi},~\eqref{eq:const_CF}, and \eqref{eq:const_tx} into Eq.~\eqref{eq:lagrangian_0}, we can then rewrite the Lagrangian function as the sum of a modified immediate cost function over the horizon $T$:

\begin{equation}
    \mathcal{L}(\pi, \lambda, \mu) = \sum_{t=1}^{T} C_{\lambda, \mu}(\mathcal{S}_t, \mathcal{A}_t) - \lambda \kappa - \mu D,
\end{equation}

\noindent where the immediate Lagrangian cost, $C_{\lambda, \mu}(\mathcal{S}_t, \mathcal{A}_t)$, serves as the per-slot objective. It is constructed by updating the base cost defined in Eq.~\eqref{eq:cost_function} with the weighted resource penalties:

\begin{equation}
    \begin{aligned}
    C_{\lambda, \mu}(\mathcal{S}_t, \mathcal{A}_t) &= \sum^{N}_{n=1}\left[\underbrace{\Delta^2_{n}(t)}_{\text{Base Cost}} \right.\\
    &\left.+ \underbrace{\lambda \xi(t) E_{\mathrm{tot},n}(\mathcal{S}_t, \mathcal{A}_t) + \mu a_{\mathrm{DEV},n}(t)C_{\mathrm{duty}}}_{\text{Resource Penalties}}\right].
    \end{aligned}
    \label{eq:lagrangian_cost}
\end{equation}

This formulation reveals the physical intuition behind the multipliers. More specifically, $\lambda$ acts as a dynamic price for \ac{cf}, inflating the cost of energy-intensive actions during high-\ac{cf} periods, while $\mu$ penalizes channel-access usage. The constant terms $\lambda \kappa$ and $\mu D$ do not affect the policy $\pi^{*}$ for fixed multipliers and can be ignored during the minimization step, which leads to the following simplified unconstrained objective:

\begin{equation}
\begin{aligned}
    \min_{\pi} \sum_{t=1}^{T}\sum_{n=1}^{N} & \left( \Delta_{n}^2(t) + \lambda \xi(t) E_{\mathrm{tot},n}(\mathcal{S}_{t}, \mathcal{A}_{t}) \right. \\
    & \left. + \mu a_{\mathrm{DEV},n}(t) C_{\mathrm{duty}} \right).
    \label{eq:unconstrained_obj}
\end{aligned}
\end{equation}

\noindent This structure is used to construct a scalable scheduling heuristic, described in the next section.


\section{The Proposed SAOITHE Framework}
\label{sec:solution}
In this section, we describe the proposed \ac{saoithe} framework, which is designed to optimize the trade-off between information freshness, i.e., \ac{aoi}, and carbon cost, i.e., \ac{cf}. We employ a Whittle-index-inspired approach~\cite{whittle1988restless}, which decomposes the intractable joint optimization into $N$ independent single-source subproblems by relaxing the hard per-slot constraints into time-averaged penalties. While standard numerical Whittle-index methods reduce complexity from exponential $O(\Delta^{N})$ to linear $O(N\cdot \Delta)$ by iteratively searching for the index, we derive a closed-form expression of the Whittle index for the relaxed single-source subproblem. This eliminates the iterative search, further reducing the complexity to $O(N \cdot \log N)$ and enabling real-time scheduling for large-scale deployments.

We derive and theoretically characterize a closed-form index for a relaxed infinite-horizon single-source model with immediate transmission and unit \ac{aoi} reset. The resulting index captures the fundamental trade-off between information staleness and transmission cost and provides a principled priority measure for feasible transmission decisions in the full finite-horizon multi-source system with buffered gateway and server dynamics.
For an $n$-th source, the relaxed cost function $\mathcal{L}_{n}(\Delta_{n},a_{\mathrm{DEV},n},t)$ at state $\Delta_{n}$ with action $a_{\mathrm{DEV},n}\in \left\{0,1\right\}$ that involves a subsidy $\nu$ for passivity, is defined as:
\begin{equation}
    \mathcal{L}_{n}(\Delta_{n}, u, t) = \Delta_{n}^{2} + a_{\mathrm{DEV},n}C_{n}(t) - (1-a_{\mathrm{DEV},n})\nu,
    \label{eq:relaxedcostfunction}
\end{equation}
where $C_{n}(t)=\lambda \xi(t)E_{\mathrm{tot},n}+\mu C_{\mathrm{duty}}$ is the corresponding immediate aggregated cost of the update.

\begin{proposition}[Closed-Form Whittle index for the Relaxed Single-Source Model]
\label{prop1}
Under the relaxed immediate-update model described above, the problem satisfies indexability, and the Whittle index $W_n(\Delta, t)$ for source $n$ at state $\Delta$ in the decoupled single-source subproblem is given by: 
\begin{equation}
    W_n(\Delta, t) = \frac{4\Delta^3 + 9\Delta^2 + 5\Delta}{6} - (\lambda \xi(t) E_{\mathrm{tot},n} + \mu C_{\mathrm{duty}}).
    \label{eq:whittle_index}
\end{equation}
\end{proposition}

The derivation involves evaluating the difference in the long-term average cost between a threshold policy $m$ and $m+1$ for a quadratic aging penalty. The detailed proof is provided in \ref{app:whittle_proof}.

The Whittle-index approach relies on indexability, i.e., the set of states in which passive actions are optimal expands monotonically as the subsidy for passivity $\nu$ increases. Indexability has been established for convex AoI cost functions in~\cite{whittle_aoi_0}. Proposition~\ref{prop1} characterizes the Whittle index for the relaxed infinite-horizon single-source model described above. For the original finite-horizon, multi-source system with buffered gateway and server dynamics, we employ this index as a scalable priority rule and evaluate the resulting policy empirically against dynamic programming and baseline schedulers.

\begin{remark}[Finite vs. Infinite Horizon Approximation]
While the optimization problem ($\mathbf{P1}$) is bounded by a finite operational horizon $T$, the Whittle index is fundamentally derived under an infinite-horizon time-averaged cost criterion. By treating the dynamic environmental cost $C_n(t)$ as locally stationary within each decision epoch (as justified in \ref{app:whittle_proof}), we apply the infinite-horizon index as a tractable heuristic approximation for the finite-horizon execution of the original multi-source problem. Because the horizon $T$ is relatively large, we expect the approximation effects to be most pronounced near the terminal boundary $t \to T$, although they may also occur earlier under tight budgets or highly nonstationary \ac{ci}. Note that we do not claim that this approximation is exact or universally optimal for all \ac{ci} traces and system configurations.
\end{remark}

\begin{corollary}[Asymptotic Stability]
\label{cor:stability}
For large $\Delta$, the Whittle index behaves asymptotically as $W_n(\Delta) \approx \frac{2}{3}\Delta^3$. Since the transmission cost $C_n(t)$ is bounded, there exists a critical age $\Delta_{\mathrm{crit}}$ such that for all $\Delta > \Delta_{\mathrm{crit}}$, the index $W_n(\Delta) > 0$. 
\end{corollary}

We evaluate the limit of the index normalized by the cubic term:
\begin{equation}
    \lim_{\Delta \to \infty} \frac{W_n(\Delta)}{\Delta^3} = \lim_{\Delta \to \infty} \frac{4\Delta^3 + 9\Delta^2 + 5\Delta - 6C_n(t)}{6\Delta^3} = \frac{2}{3}.
\end{equation}
Since the limit is a strictly positive constant, the urgency term $W_n(\Delta)$ dominates the cost term $C_n(t)$ for sufficiently large $\Delta$, forcing source activation. This shows that SAOITHE is a safe, stable control policy for critical monitoring systems, rather than just a greedy budget-saving heuristic approach.

To deploy this theoretical index within the full system, we must address the fact that the closed-form index in Eq.~\eqref{eq:whittle_index} is derived for a relaxed single-source model with immediate transmission and unit \ac{aoi} reset. In reality, gateway and server buffering introduce additional state variables and delayed, non-unit \ac{aoi} resets. To handle this, we propose a scalable heuristic: we use the idealized index purely as a source-level urgency measure to rank feasible transmissions, while the buffered gateway and server dynamics are handled explicitly by the system transitions. 

This separation preserves the computational scalability of the index-based scheduler while seamlessly accounting for multi-stage packet delivery. We outline the proposed solution in Alg.~\ref{alg:whittle_scalable}. In short, our solution decouples decision-making into four sequential steps per time slot:

\textit{1) Dynamic Cost Assessment (Step 1):} At the beginning of time slot $t$, the scheduler computes the global effective cost $C(t)$. This scalar aggregates the variable carbon penalty $\xi(t)$ and the fixed duty-cycle cost, weighted by the multipliers ($\lambda^*, \mu^*$). The real-time scheduler does not require future \ac{ci} predictions. Instead, these multipliers are computed offline via standard subgradient descent with diminishing step sizes on the dual problem over historical \ac{ci} data and then kept fixed during online execution. The resulting values are stationary shadow-price estimates for the aggregate budgets; they are not an exact certificate of primal feasibility or strong duality for every finite-horizon realization. 

\begin{algorithm}[!t]
\caption{Proposed SAOITHE approach} 
\label{alg:whittle_scalable}
\begin{algorithmic}[1]
\REQUIRE States $\mathbf{\Delta}(t)$, $\xi(t)$, Multipliers $\lambda^{\ast}, \mu^{\ast}$, Capacity $M$.
\ENSURE Activation Vector $\mathbf{a}_{\mathrm{DEV}}(t)$.

\STATE // Step 1: Compute Dynamic Cost
\STATE $C(t) \leftarrow \lambda^* \xi(t) E_{\mathrm{tot}} + \mu^* C_{\mathrm{duty}}$

\STATE // Step 2: Compute Indices (Parallel)
\FORALL{$n \in \{1,\ldots, N\}$}
    \IF{$b_{\mathrm{DEV},n}(t) = 0$}
    \STATE $W_{n}(t) \leftarrow U(\Delta_{n}) - C(t)$
    \ELSE
        \STATE $W_{n}(t) \leftarrow 0$
    \ENDIF
\ENDFOR

\STATE // Step 3: Prioritization
\STATE Sort $\mathbf{W}$ descending: $W_{(1)} \ge W_{(2)} \ge \ldots \ge W_{(N)}$
\STATE $\mathcal{A}^* \leftarrow \emptyset$
\FOR{$k = 1$ to $\min(N, M)$}
    \STATE if $W_{(k)} > 0$ then add $(k)$ to $\mathcal{A}^*$ else break
\ENDFOR

\STATE // Step 4: Execute \& Update States
\FORALL{$n \in \{1, \dots, N\}$}
    \STATE $a_{\mathrm{DEV},n}(t) \leftarrow 1$ if $n \in \mathcal{A}^*$ \textbf{and} $b_{\mathrm{DEV},n}(t) = 0$ else $0$
    \STATE Update $\Delta_n$ and channel states based on $a_{\mathrm{DEV},n}(t)$
\ENDFOR

\RETURN $\mathbf{a}_{\mathrm{DEV}}(t)$
\end{algorithmic}
\end{algorithm}

\textit{2) Parallel Index Computation (Step 2):} For every source $n \in \{1, \dots, N\}$, the algorithm determines the appropriate staleness urgency based on the current state of that source. For every idle source, the scheduler evaluates the single-source urgency $U(\Delta_n)$ and subtracts the effective transmission cost $C(t)$. The buffered gateway and server states are retained in the system evolution but are not used to alter the priority index. The Whittle index $W_n(t)$ is then obtained by subtracting the global effective cost $C(t)$ from the calculated urgency value. For busy sources ($b_{\mathrm{DEV},n}(t)=1$), the index is set to a non-positive sentinel value (here, $W_n(t)=0$) so that such sources are never selected into $\mathcal{A}^*$ in Step 3. Since these calculations rely solely on the local state component $\Delta_n$ (and the global cost $C(t)$), this step maintains $O(1)$ complexity per source and is fully parallelizable.

\textit{3) Prioritization (Step 3):} The scheduler sorts the indices in descending order to identify the top-$M$ subset $\mathcal{A}^*$. It strictly respects the channel capacity $M$ while excluding any sources with negative indices to maintain energy efficiency.

\textit{4) Execution and Update (Step 4):} 
Finally, the scheduler generates the activation vector $a_{\mathrm{DEV}}$. Consistent with the LPWAN feasibility assumptions in Section~\ref{sec:system_model}, the activation rule assigns $a_{\mathrm{DEV},n}(t)=1$ only for idle sources with $b_{\mathrm{DEV},n}(t)=0$ that belong to $\mathcal{A}^*$, ensuring that busy devices are never re-activated while a transmission is in progress. The system then updates the local states $\Delta_n$ and channel variables. Specifically, once the device activation $a_{\mathrm{DEV}}$ is determined, the forwarding and processing actions of the gateway and server deterministically follow the multi-stage buffer dynamics defined in Eqs.~\eqref{eq:h_GW_t_1} and~\eqref{eq:h_SV_t_1}.

The state space defined in Eq.~\eqref{eq:state_general} grows exponentially with the number of sources in the system. As a result, obtaining the optimal scheduling policy through an exact \ac{mdp} solution quickly becomes computationally intractable as $N$ increases. This motivates the need for scalable scheduling policies whose per-slot computational cost remains manageable even in larger systems. To demonstrate that the proposed \ac{saoithe} solution is practically deployable, we therefore analyze its computational complexity and compare it with that of optimal and existing sub-optimal approaches.

\begin{proposition}[Computational Complexity]
\label{prop:complexity}
The per-slot computational complexity of the proposed \ac{saoithe} algorithm is $O(N \cdot \log N)$. This represents a significant reduction compared to the exponential complexity, i.e., $O(\delta_{\max}^N)$, of the optimal \ac{dp} solution for the original finite-horizon problem and remains scalable with the number of sources.
\end{proposition}

The complexity analysis is provided in  \ref{app:complexity}.
In summary, although the underlying state space $\mathcal{S}_{t}$ grows exponentially with $N$, \ac{saoithe} avoids this bottleneck by using a Whittle-inspired heuristic with per-slot complexity $O(N\cdot \log N)$\footnote{The policy is not claimed to be theoretically optimal for the full buffered constrained system; its performance is evaluated empirically against the optimal \ac{dp} solution for small networks and via simulations for larger networks.}.

\section{Evaluation and Results}
\label{sec:validation}


In this section, we evaluate the performance of the proposed \ac{saoithe}. First, we compare it with the \ac{dp} solution in order to verify that the proposed scheduler achieves comparable performance. We then benchmark the proposed scheduler against two baseline policies, namely Round Robin and Random. The Round Robin policy transmits status updates periodically with a fixed period $p$, where $p$ is chosen such that the number of transmissions allowed by the carbon budget $\kappa$ under the average \ac{ci} \(\bar{\xi}(t)\) is respected, thus serving as a \ac{cf}-budget aware \ac{aoi} baseline. In contrast, the Random policy is configured to have the same average inter-transmission interval $p$, but in each time slot it performs an independent Bernoulli trial with transmission probability \(p_{\mathrm{tx}} = 1/p\), resulting in geometrically distributed inter-transmission intervals.

To reflect a practical \ac{lpwan} setting, we adopt \ac{lorawan} Class A~\cite{lorawan_spec} to model all relevant timing parameters and energy consumption. For example, standard \ac{lorawan} gateways can demodulate only a limited number of simultaneous transmissions; therefore, the number of gateway channels $M$ is set to $M=8$, which is typical for Class-A systems~\cite{adelantado2017understanding,lorawan_Mis8}. 
To model the energy consumption and timing constraints of the system, we define several physical and hardware-specific parameters as summarized in Table~\ref{tab:params}. The physical layer is configured with a bandwidth ($BW$) of $125$~kHz and a spreading factor ($SF$) of $12$, transmitting packets with a length ($L$) of 256~bits. For the server-side processing, we assume a computational capacity ($F_{\mathrm{sv}}$) of 10~GFLOP/s, where each status update requires $F_{\mathrm{task}} = 50$~MFLOP for processing and verification. The energy model is derived from the power consumption levels of the network entities: the \ac{iot} devices consume $P_{\mathrm{IoT,tx}} = 0.125$~W during transmission and $P_{\mathrm{IoT,rx}} = 0.01$~W during the two mandatory Class-A downlink windows, each lasting $100$~ms, which are counted as a marginal per-transmission cost. In contrast, the gateway receive power is absorbed into the system idle power $P_{\mathrm{idle}} = 101.5$~W, since the gateway continuously listens for all $N$ sources regardless of scheduling decisions and is therefore treated as a fixed infrastructure overhead rather than a per-packet cost. Similarly, the gateway forwarding and server processing power ($P_{\mathrm{GW,tx}}$ and $P_{\mathrm{SV,proc}}$) contribute to $E_{\mathrm{tot}}$ as active costs. Finally, the operational horizon $T$ consists of $288$ slots, each with a duration $\tau_{\mathrm{slot}} = 5$~minutes, ensuring a total evaluation period of $24$ hours.

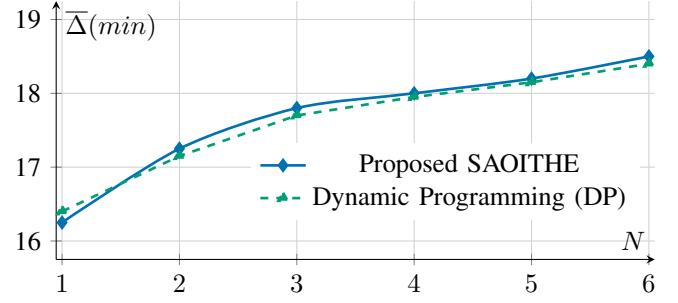
\begin{figure}[t]
\centering
	\centering
    \begin{tikzpicture}
\definecolor{plotblue}{RGB}{0, 114, 178}
\definecolor{plotorange}{RGB}{230, 159, 0}

\definecolor{plotgreen}{RGB}{0, 158, 115}
\definecolor{plotpink}{RGB}{213, 94, 0}
\definecolor{plotpurple}{RGB}{204, 121, 167}
\definecolor{plotgray}{RGB}{153, 153, 153}

\begin{axis}[
    axis lines=middle,
	ylabel= $\overline{\Delta}(min)$,
    xlabel=$N$,
    height=5.0cm,
    width=9.5cm,
	grid=both,
	major grid style={gray!35},
    xmin=0.95,
    xmax=6.05,
    ymin = 3.35,
    ymax=3.65,
    ytick={3.2, 3.4,3.5, 3.6, 3.8},
    yticklabels={16, 17, 17.5, 18, 19},
	legend style={draw=none,at={(0.65,0.95)},
    anchor=north,legend columns=1},
]











\addplot[line width=1pt, color=plotblue, mark=diamond*, smooth]table[x=N_devices,y=SAOITHE_Whittle_AoI,col sep=comma]{csv_data/dp_vs_whittle_aoi.csv};

\addlegendentry{Proposed SAOITHE}
\addplot[line width=1pt, color=plotgreen, dashed, mark=triangle*,]table[x=N_devices,y=SAOITHE_DP_AoI,col sep=comma]{csv_data/dp_vs_whittle_aoi.csv};
\addlegendentry{Dynamic Programming (DP)}









\end{axis}

\end{tikzpicture}
    \caption{Performance comparison of the exact DP solution against the proposed SAOITHE solution for
small network sizes (up to $N = 6$).}
	\label{fig:optimal_comparison}
	\vspace{-5pt}
\end{figure}

\subsection{Comparison to the DP Solution}

The \ac{dp} solution solves the considered scheduling problem by enumerating all feasible system states in the given time slot and computing the minimum cumulative cost through backward induction on the Bellman recursion. This means that, for the considered problem, the solver jointly tracks the \ac{aoi}, gateway and server busy states, buffer ages, and the associated \ac{cf} costs, then selects the action that minimizes the sum of immediate cost and the future value. However, the \ac{dp} solution becomes impractical as the state space size expands exponentially with $N$. Therefore, both memory and runtime requirements increase too quickly for realistic deployments.

\begin{figure*}[t]
	\centering
	\includegraphics[width=0.9\textwidth]{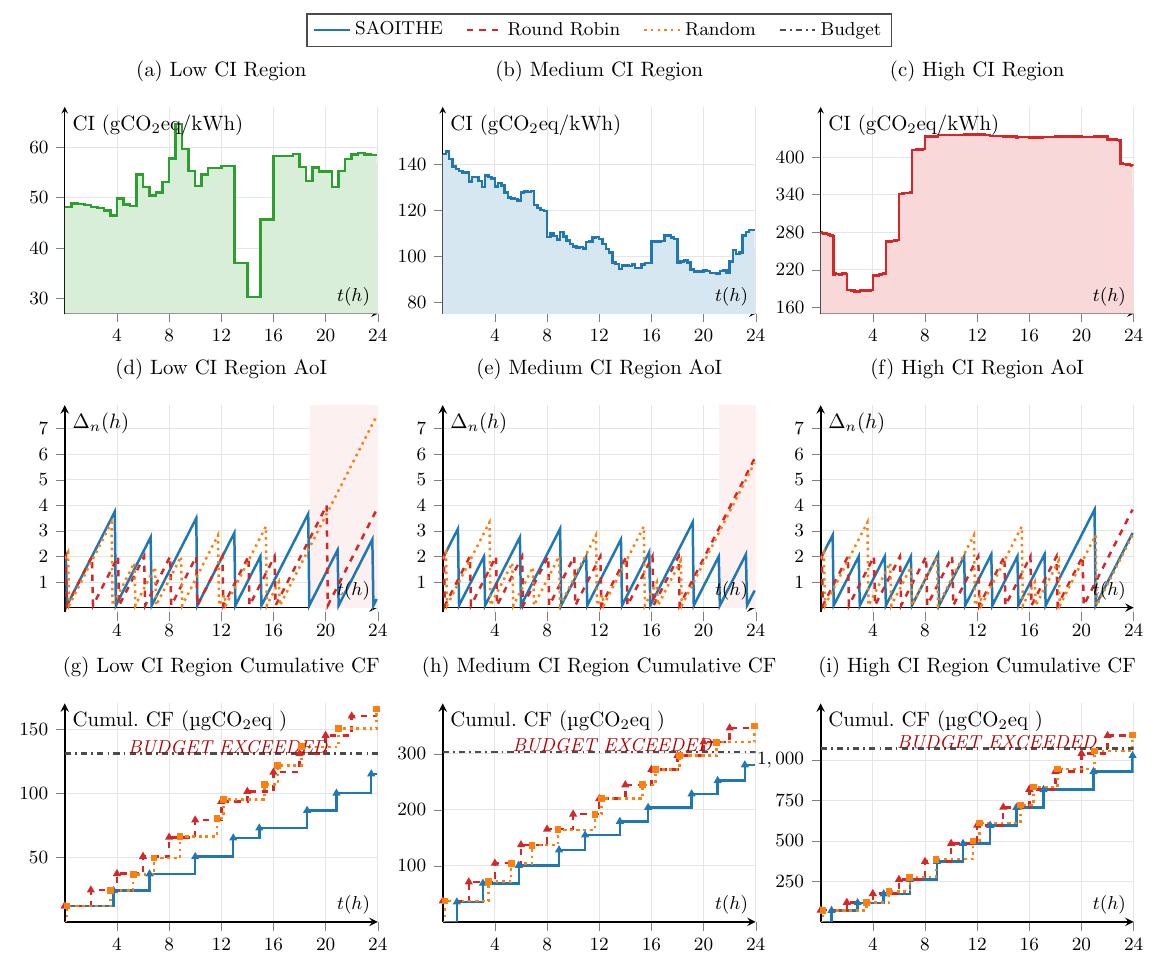}
    \caption{Comparison of CI, resulting AoI, and cumulative CF for SAOITHE, Round Robin, and Random scheduling in low, medium, and high CI regions. $N=1$}
	\label{fig:time_results}
	\vspace{-10pt}
\end{figure*}

\begin{table}[t]
\centering
\caption{Values of experimental parameters.}
\label{tab:params}
\footnotesize
\begin{tabular}{llll}
\hline
\textbf{Parameter} & \textbf{Value} & \textbf{Parameter} & \textbf{Value}  \\
\hline
BW                          & 125~kHz               & $C_{\mathrm{duty}}$          & 1296~ms              \\
$D$                         & 1\%                   & $\delta_{\max}$              & 10~slots            \\
$E_{\mathrm{tot}}$          & 0.9251~J              & $F_{sv}$                     & 10~GFLOP/s          \\
$F_{\mathrm{task}}$         & 50~MFLOP              & $\kappa$                     &\{.5, 2, 5, 21.5\}\qty{}{\micro\gram}\ce{CO2}eq \\
$L$                         & 256~bits              & $M$                          & 8                   \\
$N$                         & \{10, 50, 100\}       & $P_{\mathrm{GW,idle}}$       & 1.5~W               \\
$P_{\mathrm{GW,tx}}$        & 3~W                   & $P_{\mathrm{idle}}$          & 101.5~W             \\
$P_{\mathrm{IoT,rx}}$       & 0.010~W               & $P_{\mathrm{IoT,tx}}$        & 0.125~W             \\
$P_{\mathrm{SV,proc}}$       & 150~W                 & 
$T$                         & 288 slots (24~h) \\             SF   & 12 
    & $\tau_{\mathrm{slot}}$       & 5~min (300~s)       \\
\hline
\end{tabular}
\vspace{-15pt}
\end{table}

Fig.~\ref{fig:optimal_comparison} shows that, for network sizes up to $N=6$, the proposed \ac{saoithe} yields average \ac{aoi} values that are very close to the exact \ac{dp} solution under the same budgets. 
These results show that, for the tested small buffered systems, the proposed Whittle-based heuristic produces average \ac{aoi} values close to those obtained by \ac{dp}.

\subsection{Performance over Time and Across CI Regions}



In our work, we report average \ac{aoi} because it is the standard timeliness metric in the \ac{aoi} literature and is more easily interpretable and comparable across existing results than a quadratic \ac{aoi} cost. Although SAOITHE optimizes a cumulative quadratic \ac{aoi} objective to penalize large peaks and variability, in our experiments policies that achieve lower quadratic AoI also achieve lower average AoI, so using average \ac{aoi} as the reported metric preserves the relative ordering of schedulers and does not skew the performance comparison.

In Fig.~\ref{fig:time_results}, we illustrate how the \ac{saoithe}, Round Robin, and Random scheduling approaches evolve over time across different \ac{ci} regions. Fig.~\ref{fig:time_results}(a)--(c) show the variation of \ac{ci} over time, Fig.~\ref{fig:time_results}(d)--(f) show the corresponding evolution of \ac{aoi}, and Fig.~\ref{fig:time_results}(g)--(i) show the cumulative \ac{cf} budget consumption. From the \ac{ci} traces (Fig.~\ref{fig:time_results}(a)--(c)), we observe that, within a given region, the highest \ac{ci} values can be up to three times larger than the lowest ones, indicating high temporal variation. The \ac{aoi} results (Fig.~\ref{fig:time_results}(d)--(f)) further show that both the Random and Round Robin policies may exceed the available budget, whereas the proposed \ac{saoithe} remains within the allocated budget, as seen in Fig.~\ref{fig:time_results}(g)--(i). In all evaluated scenarios, \ac{saoithe} satisfies both the \ac{cf} constraint $J_C(\pi)\leq \kappa$ and the duty-cycle constraint $J_D(\pi)\leq D$; Round Robin and Random exceed the carbon budget in several cases (Fig.~\ref{fig:time_results}(g)--(i)).
As a result, toward the end of the operational horizon $T$, the \ac{aoi} under both Round Robin and Random increases significantly, since the system no longer has sufficient budget available to transmit new updates, while \ac{saoithe} still has budget remaining and can transmit new updates, thus keeping the \ac{aoi} lower.

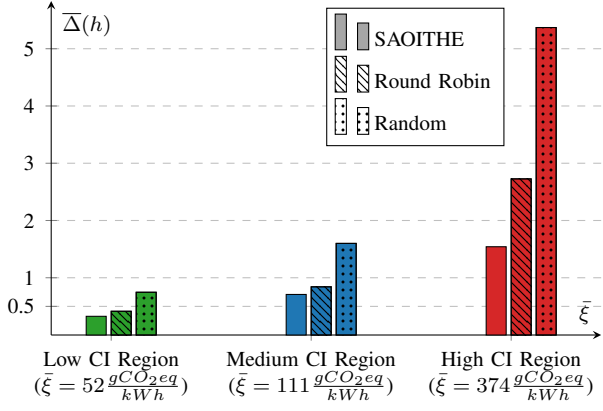
\begin{figure}
	\centering
	\large\begin{tikzpicture}
\begin{axis}[
    width=\columnwidth,
    height=6cm,
    axis lines=middle,
    ymin=0,
    ymax=70,
    ylabel={$\overline{\Delta} (h)$},
    xlabel={$\bar{\xi}$},
    xmin=0.4,
    xmax=5.9,
    ytick={6, 12,24,36,48,60},
    yticklabels={0.5, 1,2,3,4,5},
    xtick={1,3,5},
    xticklabels={
        Low CI Region\\($\bar{\xi} = 52 \frac{gCO_2eq}{kWh}$),
        Medium CI Region\\($\bar{\xi} = 111 \frac{gCO_2eq}{kWh}$),
        High CI Region\\($\bar{\xi} = 374 \frac{gCO_2eq}{kWh}$)
    },
    ymajorgrids=true,
    grid style=dashed,
    tick label style={font=\footnotesize, align=center},
    label style={font=\footnotesize},
    legend style={
    at={(0.5,0.98)},
    anchor=north west,
    draw=black,
    fill=white,
    font=\footnotesize,
    /tikz/every even column/.append style={column sep=2pt}
},
legend image post style={xscale=1.4,yscale=1.4},
    legend cell align=left
]


\addplot[
    ybar,
    bar width=0.20,
    draw=black,
    fill=NZgreen,
    forget plot
] coordinates {(0.85,3.92)};

\addplot[
    ybar,
    bar width=0.20,
    fill=NZgreen,
    forget plot
] coordinates {(1.10,4.97)};

\addplot[
    ybar,
    bar width=0.20,
    draw=black,
    fill=NZgreen!35,
    pattern=north west lines,
    pattern color=black,
    forget plot
] coordinates {(1.10,4.979288125)};

\addplot[
    ybar,
    bar width=0.20,
    fill=NZgreen,
    forget plot
] coordinates {(1.35,8.958576249)};

\addplot[
    ybar,
    bar width=0.20,
    fill=NZgreen,
    pattern=dots,
    forget plot
] coordinates {(1.35,8.958576249)};

\addplot[
    ybar,
    bar width=0.20,
    draw=black,
    fill=SIblue,
    forget plot
] coordinates {(2.85,8.5)};

\addplot[
    ybar,
    bar width=0.20,
    draw=black,
    fill=SIblue,
    forget plot
] coordinates {(3.10,10.1)};

\addplot[
    ybar,
    bar width=0.20,
    draw=black,
    fill=SIblue!70,
    pattern=north west lines,
    forget plot
] coordinates {(3.10,10.1)};

\addplot[
    ybar,
    bar width=0.20,
    draw=black,
    fill=SIblue,
    forget plot
] coordinates {(3.35,19.2)};

\addplot[
    ybar,
    bar width=0.20,
    draw=black,
    fill=SIblue!40,
    pattern=dots,
    forget plot
] coordinates {(3.35,19.2)};

\addplot[
    ybar,
    bar width=0.20,
    draw=black,
    fill=JPred,
    forget plot
] coordinates {(4.85,18.5)};

\addplot[
    ybar,
    bar width=0.20,
    draw=black,
    fill=JPred,
    forget plot
] coordinates {(5.10,32.72224274)};

\addplot[
    ybar,
    bar width=0.20,
    draw=black,
    fill=JPred!70,
    pattern=north west lines,
    forget plot
] coordinates {(5.10,32.72224274)};

\addplot[
    ybar,
    bar width=0.20,
    draw=black,
    fill=JPred,
    forget plot
] coordinates {(5.35,64.44448549)};

\addplot[
    ybar,
    bar width=0.20,
    draw=black,
    fill=JPred!40,
    pattern=dots,
    forget plot
] coordinates {(5.35,64.44448549)};

\addlegendimage{ybar,ybar legend,draw=black,fill=gray!70}
\addlegendentry{SAOITHE}

\addlegendimage{ybar,ybar legend,draw=black,fill=gray!55,pattern=north west lines}
\addlegendentry{Round Robin}

\addlegendimage{ybar,ybar legend,draw=black,fill=gray!40,pattern=dots}
\addlegendentry{Random}

\end{axis}
\end{tikzpicture}
    \caption{Average AoI for low, medium, and high CI regions  with $N=50, \kappa = $ \qty{21.5}{\micro\gram}\ce{CO2}eq.}
	\label{fig:CI_over_regions_results}
	\vspace{-10pt}
\end{figure}

\begin{figure}[t]
\centering
	\centering
    
    \begin{tikzpicture}
\definecolor{plotblue}{RGB}{0, 114, 178}
\definecolor{plotorange}{RGB}{230, 159, 0}

\definecolor{plotgreen}{RGB}{0, 158, 115}
\definecolor{plotpink}{RGB}{213, 94, 0}
\definecolor{plotpurple}{RGB}{204, 121, 167}
\definecolor{plotgray}{RGB}{153, 153, 153}

\begin{axis}[
	ymode=log,
	xmode=log,
	ylabel= $\overline{\Delta}(min)$,
    xlabel=$N$,
	grid=both,
	major grid style={gray!35},
    ytick={1, 10, 100, 200},
    yticklabels={5, 50, 500},
    ymin=0.5,
    ymax=550,
    xmin=2,
    xmax=305,
	legend style={draw=none,at={(0.4,1.3)},
    anchor=north,legend columns=3},
    ylabel style={
    rotate=-90,
    xshift=22pt,
    yshift=72pt},
]


\addlegendimage{no markers, black, line width=2pt}
\addlegendentry{SAOITHE}

\addlegendimage{no markers, plotgreen, line width=2pt}
\addlegendentry{$\kappa=$ \qty{2}{\micro\gram}\ce{CO2}eq}

\addlegendimage{area legend, fill=gray!25, draw=none}
\addlegendentry{Limit}

\addlegendimage{no markers, black, dashed, line width=2pt}
\addlegendentry{Round Robin}

\addlegendimage{no markers, plotblue, line width=2pt}
\addlegendentry{$\kappa=$ \qty{.5}{\micro\gram}\ce{CO2}eq}

\addlegendimage{no markers, white, draw=none}
\addlegendentry{}

\addlegendimage{no markers, black, dotted, line width=2pt}
\addlegendentry{Random}

\addlegendimage{no markers, plotorange, line width=2pt}
\addlegendentry{$\kappa=$ \qty{5}{\micro\gram}\ce{CO2}eq}

\addplot[line width=3pt, color=gray!25,  smooth]table[x=N,y=AoI_Gateway_Floor,col sep=comma]{csv_data/aoi_vs_N.csv};

\addplot[line width=1pt, color=plotblue, smooth]table[x=N,y=AoI_SAO_Tight_0005,col sep=comma]{csv_data/aoi_vs_N.csv};
\addplot[line width=1pt, color=plotblue, dashed, smooth]table[x=N,y=AoI_RR_Tight_0005,col sep=comma]{csv_data/aoi_vs_N.csv};

\addplot[line width=1pt, color=plotblue, dotted, smooth]table[x=N,y=AoI_Rnd_Tight_0005,col sep=comma]{csv_data/aoi_vs_N.csv};


\addplot[line width=1pt, color=plotgreen, smooth]table[x=N,y=AoI_SAO_Medium_002,col sep=comma]{csv_data/aoi_vs_N.csv};
\addplot[line width=1pt, color=plotgreen, dashed, smooth]table[x=N,y=AoI_RR_Medium_002,col sep=comma]{csv_data/aoi_vs_N.csv};

\addplot[line width=1pt, color=plotgreen, dotted, smooth]table[x=N,y=AoI_Rnd_Medium_002,col sep=comma]{csv_data/aoi_vs_N.csv};

\addplot[line width=1pt, color=plotorange, smooth]table[x=N,y=AoI_SAO_Generous_005,col sep=comma]{csv_data/aoi_vs_N.csv};
\addplot[line width=1pt, color=plotorange, dashed, smooth]table[x=N,y=AoI_RR_Generous_005,col sep=comma]{csv_data/aoi_vs_N.csv};

\addplot[line width=1pt, color=plotorange, dotted, smooth]table[x=N,y=AoI_Rnd_Generous_005,col sep=comma]{csv_data/aoi_vs_N.csv};









\end{axis}

\end{tikzpicture}
    \caption{Average AoI over the number of sources ($N$) in the system for three different CF budgets ($.5,~2, \text{and}~5$\qty{}{\micro\gram}\ce{CO2}eq) in medium CI region. }
	\label{fig:over_N_source}
	\vspace{-10pt}
\end{figure}
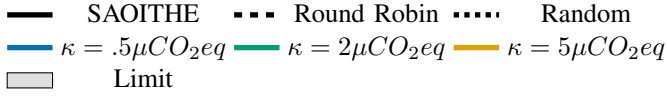

\begin{figure}
	\centering
	\begin{tikzpicture}
\definecolor{plotblue}{RGB}{0, 114, 178}
\definecolor{plotorange}{RGB}{230, 159, 0}

\definecolor{plotgreen}{RGB}{0, 158, 115}
\definecolor{plotpink}{RGB}{213, 94, 0}
\definecolor{plotpurple}{RGB}{204, 121, 167}
\definecolor{plotgray}{RGB}{153, 153, 153}

\begin{axis}[
	ymode=normal,
	xmode=log,
	ylabel= $\overline{\Delta}(h)$,
    xlabel=$\kappa(gCO_2eq)$,
	grid=both,
	major grid style={gray!35},
    ymin=0,
    ymax=140,
    ytick={0,24, 48,  72, 96, 120},
    yticklabels={0, 2, 4, 6, 8, 10},
    xmin=0.0002,
    xmax=0.055,
	legend style={draw=none,at={(0.5,1.2)},
    anchor=north,legend columns=3},
    ylabel style={
    rotate=-90,
    xshift=20pt,
    yshift=72pt},
]


\addlegendimage{no markers, black, line width=2pt}
\addlegendentry{SAOITHE}

\addlegendimage{no markers, black, dashed, line width=2pt}
\addlegendentry{Round Robin}

\addlegendimage{no markers, black, dotted, line width=2pt}
\addlegendentry{Random}

\addlegendimage{no markers, plotblue, line width=2pt}
\addlegendentry{$N=10$}

\addlegendimage{no markers, plotgreen, line width=2pt}
\addlegendentry{$N=50$}

\addlegendimage{no markers, plotorange, line width=2pt}
\addlegendentry{$N=100$}


\addplot[line width=1pt, color=plotblue,  smooth]table[x=kappa_SAOITHE_gCO2eq,y=AoI_SAOITHE_RR,col sep=comma]{csv_data/pareto_aoi_carbon_N10.csv};

\addplot[line width=1pt, color=plotblue, dashed, smooth]table[x=kappa_RR_gCO2eq,y=AoI_SAOITHE_RR,col sep=comma]{csv_data/pareto_aoi_carbon_N10.csv};

\addplot[line width=1pt, color=plotblue, dotted, smooth]table[x=kappa_RR_gCO2eq,y=AoI_Random,col sep=comma]{csv_data/pareto_aoi_carbon_N10.csv};

\addplot[line width=1pt, color=plotgreen,  smooth]table[x=kappa_SAOITHE_gCO2eq,y=AoI_SAOITHE_RR,col sep=comma]{csv_data/pareto_aoi_carbon_N50.csv};

\addplot[line width=1pt, color=plotgreen, dashed, smooth]table[x=kappa_RR_gCO2eq,y=AoI_SAOITHE_RR,col sep=comma]{csv_data/pareto_aoi_carbon_N50.csv};

\addplot[line width=1pt, color=plotgreen, dotted, smooth]table[x=kappa_RR_gCO2eq,y=AoI_Random,col sep=comma]{csv_data/pareto_aoi_carbon_N50.csv};

\addplot[line width=1pt, color=plotorange,  smooth]table[x=kappa_SAOITHE_gCO2eq,y=AoI_SAOITHE_RR,col sep=comma]{csv_data/pareto_aoi_carbon_N100.csv};

\addplot[line width=1pt, color=plotorange, dashed, smooth]table[x=kappa_RR_gCO2eq,y=AoI_SAOITHE_RR,col sep=comma]{csv_data/pareto_aoi_carbon_N100.csv};

\addplot[line width=1pt, color=plotorange, dotted, smooth]table[x=kappa_RR_gCO2eq,y=AoI_Random,col sep=comma]{csv_data/pareto_aoi_carbon_N100.csv};

\end{axis}

\end{tikzpicture}
    \caption{Average AoI as a function of the cumulative CF budget under different numbers of sources in medium CI region.}
	\label{fig:carbon_budget}
	\vspace{-10pt}
\end{figure}
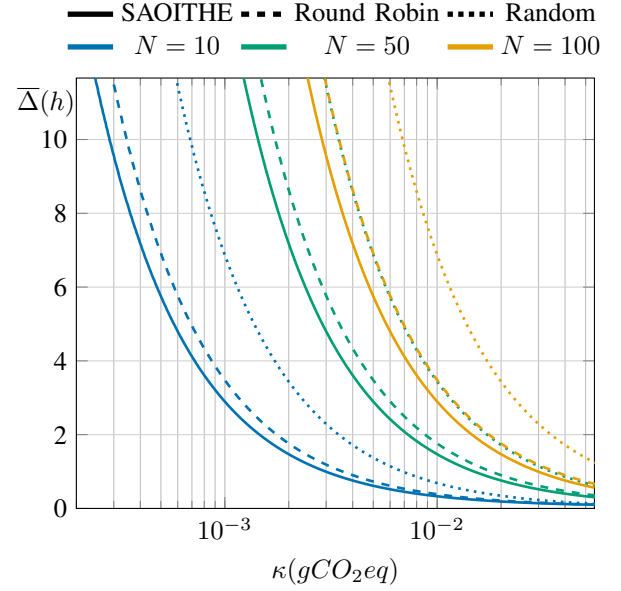

Fig.~\ref{fig:CI_over_regions_results} compares the average \ac{aoi} achieved by the three scheduling policies across low-, medium-, and high-\ac{ci} regions under the same budget $\kappa=$\qty{21.5}{\micro\gram}\ce{CO2}eq. In the low-\ac{ci} region, all policies achieve relatively small \ac{aoi} values, with \ac{saoithe} attaining an \ac{aoi} of less than $20$ minutes, Round Robin around $25$ minutes, and Random $45$ minutes. In the medium-\ac{ci} region, the gap widens, where \ac{saoithe} achieves an \ac{aoi} of $43$ minutes, while Round Robin and Random reach $50$ and $100$ minutes, respectively. The largest difference appears in the high-\ac{ci} region, where \ac{saoithe} maintains an average \ac{aoi} of one and a half hours, whereas Round Robin increases to $2$ hours $45$ minutes and Random to over five and a half hours. These results show that the performance gap between the proposed \ac{saoithe} and the baseline policies becomes larger as the average \ac{ci} increases. In particular, while all approaches experience higher \ac{aoi} in more carbon-intensive regions, \ac{saoithe} degrades much more gracefully, since it explicitly accounts for the carbon cost of transmissions when allocating the available budget. In contrast, Round Robin, and especially Random, fail to adapt their transmission decisions to the \ac{ci}, which leads to inefficient use of the available budget and substantially higher \ac{aoi}. Overall, the results confirm that the advantage of employing \ac{saoithe} is lower in low- and medium-\ac{ci} regions, where the gains are around $25\%$ and $20\%$, respectively, but becomes more advantageous in high-\ac{ci} settings, where such gains can be almost doubled, up to $75\%$, compared with Round Robin.

\subsection{Scalability and Whittle-Boundary Analysis}

To demonstrate the scalability of the proposed solution, Fig.~\ref{fig:over_N_source} evaluates the proposed approach and two baselines under three different budgets, namely $\kappa=$ \qty{0.5}{\micro\gram}\ce{CO2}eq, $\kappa=$ \qty{2}{\micro\gram}\ce{CO2}eq, and $\kappa=$ \qty{5}{\micro\gram}\ce{CO2}eq, respectively. As the budget increases, all solutions approach the physical limit imposed by the employed \ac{lorawan} technology. However, regardless of the budget, the proposed \ac{saoithe} achieves lower \ac{aoi} in comparison. Furthermore, as the number of sources increases, the difference between the approaches becomes larger, as the \ac{aoi} also increases. The result is most notable when the budget is high, as the proposed solution operates near the hardware limit of the system, while the performance of Round Robin deteriorates with the increasing number of sources. Consequently, the \ac{saoithe} solution is scalable, since the total number of transmissions per horizon is capped by the gateway demodulation capacity ($M=8$), independently of $N$. Additionally, this reveals a fundamental scalability paradigm: once the gateway capacity is saturated, the network's maximum \ac{cf} reaches a physical plateau, meaning that a higher budget will not result in lower \ac{aoi}.





Fig.~\ref{fig:carbon_budget} illustrates the trade-off between information freshness and the available \ac{cf} budget $\kappa$ for systems with different numbers of information sources, i.e., $N=10$, $N=50$, and $N=100$. As the number of sources increases, a larger budget is required to achieve the same \ac{aoi}, which is expected due to the increased contention among sources. Across all considered system sizes, \ac{saoithe} consistently achieves the lowest \ac{aoi} for a given budget, thereby providing the most efficient freshness-carbon trade-off. In contrast, Round Robin performs worse, while Random exhibits the weakest performance and requires a substantially larger \ac{cf} budget to achieve the same \ac{aoi} levels. The results further show that the overall \ac{aoi} increases with the number of sources, highlighting the greater scheduling complexity in larger systems, while preserving the relative advantage of \ac{saoithe}. Moreover, as the available budget increases, the gap in \ac{aoi} among the considered approaches decreases, as expected, since the system becomes less constrained under a larger budget.

 In Fig.~\ref{fig:whittle_index}, we illustrate the theoretical decision boundary of the \ac{saoithe} solution by tracing $\Delta_{\mathrm{crit}}(\xi)$ (see Corollary~\ref{cor:stability}) across the \ac{ci} range $\xi \in [50, 450]$. The boundary partitions the state space into a Transmit Region ($W_{n} > 0$, $\Delta > \Delta_{\mathrm{crit}}$), where information staleness outweighs the carbon cost of transmission, and a Wait/Buffer Region ($W_n \leq 0$), where the scheduler defers transmissions to preserve the carbon budget. The cube-root non-linearity $\Delta_{\mathrm{crit}} \propto \xi(t)^{1/3}$ confirms that the policy becomes increasingly conservative as the grid becomes more carbon-intensive. More specifically, the critical \ac{aoi} $\Delta_{\mathrm{crit}}$ rises from approximately $24$ minutes at $\xi = 50~\mathrm{gCO_{2}/kWh}$ to over $50$ minutes at $\xi = 450~\mathrm{gCO_{2}/kWh}$. The boundary is governed by $\lambda^{*} E_{\mathrm{tot}}$, where $E_{\mathrm{tot}} = 0.9251$~J and $\lambda^{*}$ is the Lagrange multiplier that emerges at the convergence of the SAOITHE algorithm, representing the shadow price of the carbon budget $\kappa$. Notably, the boundary is defined per device. In other words, the Whittle decomposition reduces the $N$-device scheduling problem to $N$ independent single-device subproblems. As a result, $\Delta_{\mathrm{crit}}(\xi)$ applies uniformly to all devices regardless of network size. Furthermore, the curve implicitly encodes a specific $\kappa$: a tighter carbon budget drives $\lambda^{*}$ higher, shifting $\Delta_{\mathrm{crit}}$ upward and forcing the scheduler to defer transmissions more aggressively, while a relaxed budget lowers $\lambda^{*}$ and compresses the curve toward smaller critical ages.

Overall, the results confirm that \ac{saoithe} provides an effective and scalable solution for carbon-aware status updating in hardware-constrained edge networks. The proposed Whittle-index-based scheduler consistently achieves lower \ac{aoi} than the baseline policies while respecting \ac{cf} and duty-cycle budgets across diverse \ac{ci} conditions. Its closed-form index provides analytically informed source prioritization, while the explicit gateway and server state transitions capture buffering, delayed delivery, and multi-stage packet processing in the full finite-horizon system. Comparisons with \ac{dp} for small networks demonstrate close performance to the exact solution, and the resulting low-complexity policy supports scalable operation for larger deployments. Feasibility is demonstrated using realistic \ac{ci} traces, while the LoRaWAN-based model captures key timing and energy characteristics relevant to practical deployments. Building on this foundation, future work will extend the framework to heterogeneous source priorities, weighted \ac{aoi}, source-dependent \ac{ci}, and \ac{swipt}, enabling further refinement of the trade-offs among information freshness, \ac{cf}, and system constraints.
 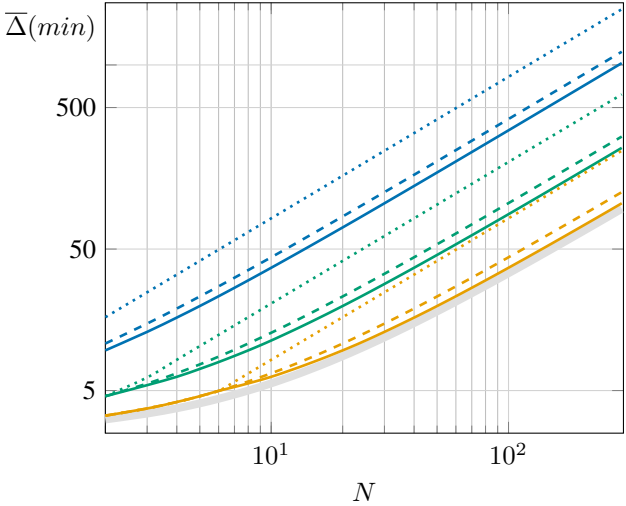
\begin{figure}
	\centering
	\large\begin{tikzpicture}
\begin{axis}[
    axis lines=middle,
    width=\columnwidth,
    height=6.2cm,
    xmin=50,
    xmax=450,
    ymin=0,
    ymax=11.2,
    xlabel={$\xi(t)$ (gCO$_2$/kWh)},
    ylabel={$\Delta_{\mathrm{crit}} (min)$ },
    ytick={2,4,6,8,10},
    yticklabels={10,20,30,40,50},
    grid=both,
    major grid style={dashed, gray!60},
    minor grid style={dashed, gray!30},
    tick label style={font=\footnotesize},
    label style={font=\footnotesize},
    legend style={
        at={(0.98,0.15)},
        anchor=south east,
        draw=gray!70,
        fill=white,
        font=\footnotesize
    },
    legend cell align=left,
    axis on top
]

\addplot[
    black,
    thick
] table[
    col sep=comma,
    x index=0,
    y index=1
] {csv_data/whittle_decision_boundary.csv};
\addlegendentry{$W_n = 0$ boundary}

\addlegendimage{area legend, fill=red!20, draw=none}
\addlegendentry{Transmit Region ($W_n > 0$)}

\addlegendimage{area legend, fill=blue!25, draw=none}
\addlegendentry{Wait/Buffer Region ($W_n \leq 0$)}

\addplot[
    name path=boundary,
    draw=none
] table[
    col sep=comma,
    x index=0,
    y index=1
] {csv_data/whittle_decision_boundary.csv};

\path[name path=topline] (axis cs:50,11.2) -- (axis cs:450,11.2);
\path[name path=bottomline] (axis cs:50,0) -- (axis cs:450,0);

\addplot[
    blue!25,
    draw=none
] fill between[of=boundary and bottomline];

\addplot[
    red!20,
    draw=none
] fill between[of=boundary and topline];



\end{axis}
\end{tikzpicture}
    \caption{Whittle decision boundary of SAOITHE as a function of CI $\xi(t) \in [50, 450]$~gCO$_2$/kWh, with $E_{tot} = 0.9251$~J. The boundary $\Delta_{crit}(\xi)$ separates the transmit region ($W_n > 0$, in red) from the wait region ($W_n \leq 0$, in blue).}
	\label{fig:whittle_index}
	\vspace{-10pt}
\end{figure}

\section{Conclusion}
\label{sec:conclusion}

In this paper, we proposed the \ac{saoithe} framework for \ac{iot} devices, which minimizes a quadratic \ac{aoi} penalty subject to an allocated \ac{cf} budget, transmission duty usage, and channel capacity constraints. We showed that the proposed Whittle-index-based solution is scalable and outperforms Round Robin and Random baseline solutions. Using real-world \ac{ci} traces, our results show a fundamental trade-off between information freshness and the resulting \ac{cf} of the system. Additionally, we demonstrate that \ac{saoithe} consistently provides the best \ac{aoi}-\ac{cf} trade-off while remaining scalable and practically viable. Furthermore, the results confirm that the advantage of employing \ac{saoithe} is smaller in low- and medium-\ac{ci} regions, where the gains are around $25\%$ and $20\%$, respectively, but becomes more pronounced in high-\ac{ci} settings, where such gains can be almost doubled, reaching up to $75\%$ compared with the Round Robin baseline.

\section*{Acknowledgements}
This work was supported by the Slovenian Research and Innovation Agency
under grants P2-0016, MN-0009, and J2-50071.

\section*{Declaration of Generative AI and AI-assisted technologies in the writing process}

During the preparation of this work, the authors used ChatGPT and Perplexity to improve the grammar and language of certain parts of the manuscript.



 \bibliographystyle{elsarticle-num} 
 \bibliography{cas-refs}
\section*{Appendices}
\label{sec:appendix}
\appendix
\section{Derivation of the Closed-Form Whittle Index}
\label{app:whittle_proof}

In real-time energy markets, the \ac{ci} ($\xi(t)$) typically updates in discrete intervals, such as every $5$ minutes (corresponding to our decision slot $\tau_{\mathrm{slot}} = 300$~s). In contrast, the actual physical communication, taking \ac{lorawan} packet transmission as an example, takes only a few seconds. Because the transmission time is orders of magnitude smaller than the \ac{ci} update interval, the environmental cost $C_n(t)$ remains effectively constant during the execution of any scheduling decision. Therefore, the dynamic cost can be accurately treated as locally stationary at each decision epoch.

By assuming a locally stationary cost $C_n(t)$ for the duration of the decision interval, we consider the decoupled single-source problem with the relaxed Lagrangian cost function:
\begin{equation}
\mathcal{L}_{n}(\Delta_{n},a,t)=\Delta_{n}^{2} + aC_{n}(t)-(1-a)\nu.
\label{eq:relaxed_lagrangian}
\end{equation}
Here, $\nu$ is a subsidy for passivity: passive slots incur the reduced cost $\Delta_{n}^{2}-\nu$, while active slots incur the cost $\Delta_{n}^{2}+C_{n}(t)$. Thus, $\nu$ shifts the relative cost of transmitting versus waiting when computing the Whittle index. Since the \ac{aoi} penalty ($\Delta^2$) is convex, the optimal strategy for the relaxed single-source problem is a threshold policy: transmit if the current age $\Delta \ge H$. Under a threshold $H$, the system evolves in renewal cycles of length $H$: the source remains passive for $H-1$ slots and transmits at age $H$, after which the AoI is reset to one.

The cumulative staleness cost over one cycle is
\begin{equation}
J_{\Delta}(H) = \sum_{h=1}^{H} h^{2} = \frac{H(H+1)(2H+1)}{6}.
\end{equation}
In each cycle, there are $H-1$ passive slots and one active slot. The cumulative Lagrangian cost over one cycle is therefore
\begin{equation}
J_{\mathcal{L}}(H) = J_{\Delta}(H) - (H-1)\nu + C_n(t).
\end{equation}
The long-term average cost $\Phi(H)$ is obtained by amortizing the cycle cost over the cycle length:
\begin{equation}
    \Phi(H) = \frac{1}{H} \left[ \frac{H(H+1)(2H+1)}{6} - (H-1)\nu + C_{n}(t) \right].
    \label{eq:avg_lagrangian}
\end{equation}

The Whittle Index $W_n(\Delta,t)$ is defined as the critical subsidy $\nu$ that makes the decision maker indifferent between transmitting at age $\Delta$ (threshold policy $H=\Delta$) and waiting one more slot (threshold policy $H=\Delta+1$). Equating the average costs $\Phi(\Delta)$ and $\Phi(\Delta+1)$ yields
\begin{equation}
    \frac{J_{\Delta}(\Delta) - (\Delta-1)\nu + C_n(t)}{\Delta}
    =
    \frac{J_{\Delta}(\Delta+1) - \Delta\nu + C_n(t)}{\Delta+1}.
    \label{eq:indifference}
\end{equation}
Solving \eqref{eq:indifference} for $\nu$ gives the critical subsidy:
\begin{align}
    \nu
    &= \Delta J_{\Delta}(\Delta+1) - (\Delta+1) J_{\Delta}(\Delta) - C_n(t).
    \label{eq:nu_intermediate}
\end{align}
Substituting the cubic expression $J_{\Delta}(H) = H(H+1)(2H+1)/6$ into \eqref{eq:nu_intermediate} and simplifying yields
\begin{align}
    \nu
    &= \frac{4\Delta^3 + 9\Delta^2 + 5\Delta}{6} - C_n(t).
\end{align}
For simplicity, we denote the staleness urgency function as $U(\Delta_n) = \frac{4\Delta_n^3 + 9\Delta_n^2 + 5\Delta_n}{6}.$ 

By definition, the Whittle index equals this critical subsidy, so we obtain the closed-form index
\begin{equation}
    W_{n}(\Delta, t) = U(\Delta_{n}) - C_n(t),
    \label{eq:whittle_index_final}
\end{equation}
where $C_n(t) = \lambda \xi(t) E_{\mathrm{tot},n} + \mu C_{\mathrm{duty}}$.

\begin{remark}[Numerical verification]
We verified the closed-form index \eqref{eq:whittle_index_final} against value iteration for the relaxed single-source MDP over a wide range of ages and cost parameters. For each age $\Delta$, the subsidy $\nu$ that equalizes the average cost of thresholds $H=\Delta$ and $H=\Delta+1$ matches the analytical expression up to numerical precision, confirming consistency between the derivation and the numerical solution.
\end{remark}


\section{Complexity Analysis of \ac{saoithe}-Whittle}
\label{app:complexity}

In this appendix, we analyze the asymptotic computational complexity of the proposed scheduling algorithm. Let $\mathcal{N} = \{1, \dots, N\}$ denote the set of devices managed by the gateway. The scheduling decision at each time slot $t$ consists of two sequential phases: Index Computation and Prioritization.

The gateway computes the Whittle Index $W_n(\Delta, t)$ for all $n \in \mathcal{N}$ using the closed-form expression derived in Eq.~\eqref{eq:whittle_index}. Let $\mathcal{T}_{\mathrm{op}}$ denote the computational cost of a basic floating-point operation (addition, multiplication). The calculation of a cubic polynomial $W_n \approx c_3\Delta^3 + c_2\Delta^2 + c_1\Delta + c_0$ requires a fixed constant number of operations, $k$, regardless of the network size $N$.
Thus, the total cost for index computation, $T_{\text{calc}}$, scales linearly with $N$:
\begin{equation}
    T_{\text{calc}}(N) = \sum_{n=1}^{N} k \cdot \mathcal{T}_{\mathrm{op}} = O(N).
\end{equation}

\noindent Next, the scheduler must select the subset of devices $\mathcal{A}^* \subset \mathcal{N}$ with cardinality $|\mathcal{A}^*| \leq M$ that maximizes the sum of indices. This is equivalent to sorting the set $\{W_n\}_{n=1}^N$ in descending order.
According to the information-theoretic lower bound for comparison-based sorting algorithms~\cite{cormen01introduction}, the worst-case computational complexity $T_{\text{sort}}$ is:
\begin{equation}
    T_{\text{sort}}(N) = O(N \cdot \log N).
\end{equation}

\noindent As a result, the total computational complexity per control slot is the sum of the two phases:
\begin{equation}
    T_{\text{total}}(N) = T_{\text{calc}}(N) + T_{\text{sort}}(N) = O(N) + O(N \cdot\log N).
\end{equation}

\noindent Since $N \log N \gg N$ when $N\rightarrow\infty $, the asymptotic complexity is:

\begin{equation}
    T_{\text{total}}(N) \approx  O(N \cdot\log N).
\end{equation}

\noindent This confirms that the algorithm scales quasilinearly, making it computationally feasible for resource-constrained gateways even as $N \to \infty$.

\end{document}